\documentclass[a4paper]{article}

\usepackage{a4wide}
\usepackage{graphicx,subfigure,hhline}
	\graphicspath{{./}{figures/}}
\usepackage[colorlinks=true,linkcolor=blue,citecolor=red]{hyperref}
\usepackage{xcolor}
\usepackage{amsfonts,amsmath,amssymb,amsthm,mathrsfs,mathtools,bbold,empheq}
\usepackage[inline]{enumitem}
\usepackage[ruled,vlined,linesnumbered]{algorithm2e}
\usepackage{caption} 
\usepackage{geometry}
\theoremstyle{plain} 
\newtheorem{theorem}{Theorem} 

\theoremstyle{remark}\newtheorem{remark}{Remark}[section]

\newcommand{\R}{\mathbb{R}}

\newcommand{\pr}[1]{{}^\prime\!#1}

\usepackage{authblk}

\allowdisplaybreaks

\begin{document}
\title{An SIR model with viral load--dependent transmission}
\author[1]{Rossella Della Marca\thanks{\texttt{rossella.dellamarca@sissa.it}}}
\author[2]{Nadia Loy\thanks{Corresponding author, \texttt{nadia.loy@polito.it}}}
\author[2]{Andrea Tosin\thanks{\texttt{andrea.tosin@polito.it}}}
\affil[1]{{\small Mathematics Area, SISSA -- International School for Advanced Studies, Trieste, Italy}}
\affil[2]{{\small Department of Mathematical Sciences ``G. L. Lagrange'', Politecnico di Torino, Torino, Italy}}
\date{}
\maketitle

\begin{abstract}
The viral load is known to be a chief predictor of the risk of transmission of infectious diseases. In this work, we investigate the role of the individuals' viral load in the disease transmission by proposing a new susceptible--infectious--recovered epidemic model for the densities and mean viral loads of each compartment. To this aim, we formally derive the compartmental model from an appropriate microscopic one. Firstly, we consider a multi--agent system in which  individuals  are identified by the epidemiological compartment to which they belong and by their viral load. Microscopic rules describe both the switch of  compartment and the evolution of the viral load. In particular,  in the binary interactions between susceptible  and infectious individuals, the probability for the susceptible individual to get infected depends on the viral load of the infectious individual. Then, we implement the prescribed microscopic dynamics in appropriate  kinetic equations, from which the macroscopic equations for the densities and viral load momentum of the compartments are eventually {derived}. In the macroscopic model, the rate of disease transmission turns out to be a function of the mean viral load of the infectious population. We analytically and numerically investigate  the case that the  transmission rate linearly depends on the viral load, which is compared to the classical case of constant  transmission rate. A qualitative analysis is performed based on stability and bifurcation theory. Finally,  numerical investigations concerning the model reproduction number and the epidemic dynamics are presented.
\end{abstract}

\smallskip

\noindent{\bf Keywords:} Boltzmann--type equations, Markov--type jump processes, epidemic,  basic reproduction number, viral load, qualitative analysis

\smallskip

\noindent{\bf Mathematics Subject Classification:} 35Q20, 35Q70, 35Q84, 37N25%Dynamical systems in biology

\section{Introduction}
The route of transmission of many infectious diseases is given by social contacts among  individuals. The virus shed by an infectious individual  may be transmitted to a healthy one during an encounter, so that the disease also develops in the latter. There is evidence that the quantity of  virus carried by the infectious individual determines the occurrence or not of the transmission: as it is reasonable to expect, higher is the viral load of the infectious individual higher is the probability of transmitting the infection. For example, the quanta emission rate (ER$_\text{q}$) measures the number of quanta (a quantum is the dose of airborne droplet nuclei required to cause infection in 63\% of susceptible persons) into the air per time unit. The ER$_\text{q}$ for respiratory diseases (including SARS--CoV--1, SARS--CoV--2, MERS, measles virus, adenovirus, rhinovirus, coxsackievirus, seasonal influenza virus and \textit{Mycobacterium tuberculosis}) has been estimated as directly [resp. inversely] proportional to the  viral load in sputum [resp. the infectious dose] \cite{MIKSZEWSKI2021}: a more contagious strain would have higher ER$_\text{q}$ values through a higher median viral load and/or a lower infectious dose. The viral load is also the chief predictor of the risk of sexually--transmitted infections, like HIV/AIDS \cite{quinn,wilson2008}.

In the mathematical epidemiology community, the awareness of the importance of the viral load in the  dynamics of infectious diseases has recently led to the development of epidemic models that explicitly incorporate such a microscopic trait \cite{volpertimmuno,DMrLnTa22,loy2021KRM,loy2021MBE}. Specifically, in the paper \cite{loy2021KRM} the authors propose a modelling framework through kinetic equations in which  individuals  are characterized by a discrete label and by their viral  load; then, a prototype epidemic model is introduced in order to  illustrate the impact of individuals' viral  load on
test--and--isolate activities. This work is extended in the paper \cite{loy2021MBE}, where the authors propose a kinetic model for the spread of an infectious disease on a graph, the nodes here representing different spatial locations. By following the wake of papers \cite{loy2021KRM,loy2021MBE}, in the paper \cite{DMrLnTa22} the authors introduce a compartmental susceptible--infectious--isolated--recovered model, in which the individual viral load  evolves according to appropriate microscopic rules and determines the probability of isolation of infectious individuals. 

To the best of our knowledge, the first epidemic model that incorporates the role of viral load in the disease \textit{transmission} term has been  proposed  by Banerjee \textit{et al}. \cite{volpertimmuno}.  In the paper \cite{volpertimmuno} an immuno--epidemiological model is introduced, where the number of susceptible people depends on the number of infectious people through the initial viral load acquired during the interactions. More precisely, according to the model \cite{volpertimmuno}, the growth in the number of infectious individuals increases the initial viral load, and
provides a \textit{switch} from the first stage of the epidemic where only people with weak immune response can be
infected to the second stage where also people with strong immune response can be infected. 

 In the present work, we investigate how the viral load of infectious individuals affects the probability of disease {transmission}, and the consequent epidemic dynamics,    by relying on the modelling framework of \textit{kinetic equations} for \textit{multi--agent systems}. Kinetic theory and Boltzmann--like equations have proved to be a very effective tool to enhance the description of infectious diseases dynamics,  by allowing the incorporation in the model of not only the role of viral load \cite{DMrLnTa22,loy2021KRM,loy2021MBE}, but also  that of: social structure and wealth distribution within the host population \cite{Bernardi,DimarcoPareschi2020,DMPerJOMB,ZanellaBardelliAzzi}, contact heterogeneity \cite{Dimarco2021,MedZan2021}, implementation of lockdown measures \cite{Lockdown} and spatial propagation of the infection \cite{Bertaglia2,BertagliaPareschi,loy2021MBE}. Specifically, we follow the approach of the papers \cite{DMrLnTa22,loy2021KRM,loy2021MBE}, by starting from a detailed description of the microscopic dynamics of the  disease spread,  microscopic dynamics that is shared  by all the individuals (also called the \textit{agents}) that are assumed to be indistinguishable. Then,  we introduce suitable kinetic equations that give a statistical portrait of the agents of the system by following exactly the prescribed microscopic rules. Eventually, from the kinetic equations we derive a macroscopic model for the aggregate description of the system that  naturally inherits the details of the microscopic dynamics.
 
We assume that 
 the individuals are characterized by a double \textit{microscopic state}:
   a \textit{label}, that denotes the epidemiological compartment to which they belong, and  a \textit{physical quantity} that is chosen to be the individual \textit{viral load}. The microscopic dynamics is described in terms of \textit{microscopic interactions} that allow the viral load to evolve and by the means of Markovian processes ruling the switch of compartment. 
We consider a basic susceptible--infectious--recovered (SIR) compartmental structure and assume that  the mechanism of disease {transmission}  (leading the healthy individuals to become ill)   depends on the viral load of the infectious individuals. 

The rest of the manuscript is organized as follows. In Section \ref{sec:2}, we present our multi--agent system and the microscopic dynamics. Then, we revise the modelling framework proposed in the paper \cite{DMrLnTa22} in order to derive the macroscopic compartmental model including the role of viral load in the rate of disease transmission. In Section \ref{sec:4}, we perform a qualitative analysis of the proposed model by determining the equilibria and investigating their stability in terms of the basic reproduction number $\mathcal{R}_0$. In Section \ref{sec:5}, some numerical simulations of the macroscopic model are performed: both the reproduction number and the epidemic temporal dynamics under the assumption of viral load--dependent rate of disease transmission are compared with those under  the \textit{classical} assumption of constant rate of disease transmission. Finally, in Section \ref{sec:6}, we draw some conclusions. 

\section{The mathematical model: from a multi--agent system to compartmental macroscopic equations}\label{sec:2}
Let us consider an infectious disease spreading among individuals as a consequence of social contacts. Individuals are  modelled as agents of a {multi--agent system} and characterized by a microscopic state. In particular, the agents are divided into disjoint {compartments} depending on their  state of health with respect to the disease. Moreover, they are characterized by a physical quantity named \textit{viral load}, that represents the quantity of viral particles present in the organism.

\subsection{The microscopic model} 
At any time $t$ each agent of the system  is characterized by a microscopic state $(x,v)$, where $x \in \mathcal{X}$ is a label that takes into account the  epidemiological compartment to which the agent belongs, and $v\in[0,1]$ is a normalized measure of the individual's viral load, being $v=1$ the  maximum observable value.

The evolution of both the label and the viral  load may be described by the means of microscopic stochastic processes, that can be expressed through Markovian processes \cite{DMrLnTa22}, namely through transition probabilities
\begin{equation*}
	P\left( (j,v) \rightarrow (i,v')\right),
\end{equation*}   
that is the conditional probability for an agent to change microscopic state from $(j,v)$ to $(i,v')$, with $(j,v),\,(i,v') \in \mathcal{X}\times [0,1]$. In general,  the viral load  of an individual may change both simultaneously to and independently of the switch of  compartment. In the second case, we consider transition probabilities for the mere evolution of the viral load of an individual in class $i\in\mathcal{X}$ that we denote by
\begin{equation*}
	P_i(v\rightarrow v')= P\left( (i,v) \rightarrow (i,v')\right),
\end{equation*}  
with $v,\,v'\in[0,1]$.

{Instead, if only the compartment changes, then we denote by  $P(i\rightarrow j)$ the probability for an agent to switch from the compartment $i$ to the compartment $j$.}

\subsubsection{The compartmental structure}
At any time $t$ the agents, labelled with $x\in \mathcal{X}$, are divided in the following disjoint epidemiological  compartments:
\begin{itemize}
	\item \textit{susceptible}, $x=S$:  individuals who are healthy but can contract the disease. The susceptible population
	increases by a net inflow, incorporating  new births and immigration, and decreases due to disease transmission and natural death;
	\item \textit{infectious}, $I$: individuals who  are infected by the disease and can transmit the virus to others. 
	Infectious individuals arise as the result of new infections of susceptible individuals and diminish due to recovery  and natural death; 
	\item  \textit{recovered}, $x=R$:  individuals who have recovered from the disease after the infectious period. They come from the infectious compartments $I$
	and acquire long lasting immunity against the disease. Recovered people diminish only due to natural death.
\end{itemize}
Specifically: susceptible individuals have $v\equiv0$; once infected, an individual's  viral load increases until reaching a peak value (that varies from person to person) and then gradually decreases, see e.g. the  representative plot of SARS--CoV--2 viral load evolution given in the paper \cite{cevik2020virology}, Fig. 2. Hence, for mathematical convenience \cite{DMrLnTa22}, we assume that members of the class $I$
are further divided into:
\begin{itemize}
	\item  infectious 
	with increasing viral load, $x=I_1$;
	\item infectious
	with decreasing viral load, $x=I_2$.
\end{itemize}
Note that new infections enter the class $I_1$, while recovery may occur only during the stage $I_2$. Finally, after the infectious period, recovered individuals may still have a positive viral load  which however definitively  approaches zero, as live virus could no longer be cultured
(see e.g. the studies \cite{cevik2020virology,He2020} on COVID--19 viral shedding).

Also, since our model incorporates \textit{birth} and \textit{death} processes, we  introduce the following two auxiliary compartments: individuals that enter the susceptible class by newborn or immigration, $x=B$, and individuals who die of natural causes, $x=D$. We assume that members of the class $B$ have $v\equiv 0$, while those of the class $D$
retain the viral load value at the time they died. Individuals can switch from the class $B$ to the class $S$ with frequency $\lambda_b$ and probability $P(B\rightarrow S)= b/\rho_B(t)$. The quantity $\rho_B(t)$, that will be defined later, measures the size  of the class $B$  at time $t$.  Moreover, all the living individuals can die, thus moving to the class $x=D$, with frequency $\lambda_{\mu}$ and probability $P(i\rightarrow D)=\mu$, being $i\in\{S,I_1,I_2,R\}$.

\subsubsection{Evolution of the viral  load}\label{Sec:viral-load}
Let us now focus on the mathematical modelling of the evolution of an individual's viral load $v$. 
We distinguish  the two following cases when $v$ changes over time: i) a susceptible individual, having $v=0$, acquires a positive viral load (and gets infected) by interaction with an infectious individual; ii) the viral loads of infectious ($I_1$, $I_2$) 
and recovered ($R$) individuals evolve naturally in virtue of physiological processes. 

Given an agent labelled with $S$, then the necessary condition for acquiring a positive viral load is an encounter with an infectious agent ($I_1$ or $I_2$). {Therefore, we model the disease transmission process as a \textit{binary interaction}}, thus relying on the typical tools of kinetic theory \cite{pareschi2013BOOK}.  Let us denote by $\lambda_\beta>0$ the frequency of these interactions. Increasing [resp. decreasing] $\lambda_\beta$ corresponds to increasing [resp. reducing] encounters among people: the lower $\lambda_\beta$ the more strengthened social distancing.  

By interacting with an infectious individual carrying viral load $w>0$, a susceptible individual does or does not get infected. In the first case his/her  viral load after the interaction (say, $v'$) is positive: $v'>0$; in the second case it remains null: $v'=0$. Specifically, we consider the following {microscopic binary interaction rule}:
\begin{equation*}%\label{eq:micro_bin}
	v'=T_{\nu_{\beta}}v_0,\quad w'=w,
\end{equation*}
where $T_{\nu_{\beta}}$ is a Bernoulli random variable of parameter $\nu_{\beta}={\nu_{\beta} (w)}\in (0,1)$  describing the case of successful transmission of the disease ($T_{\nu_{\beta}}=1$) and the case of contact without transmission ($T_{\nu_{\beta}}=0$). It seems us reasonable to assume that $\nu_{\beta}(w)$, that we name  \textit{transmission function}, is a non--decreasing function of $w$, the viral load of the infectious individual. 

We assume that new infected individuals enter the  class $I_1$ and they all acquire the same initial viral load, $v_0$ (that can be interpreted as an average initial value). 
We remark that this binary interaction process causes simultaneously a change of the microscopic state $v$ and a label switch, because as soon as $v$ becomes positive, i.e. if $T_{\nu_\beta}=1$, the susceptible individual switches to the class $I_1$. 
In terms of transition probabilities for the susceptible individual, this can be expressed as
\begin{equation*}
	P((S,v) \rightarrow (I_1,v'))= \nu_\beta(w)P_S(v \rightarrow v'), \quad  P_S(v \rightarrow v')=\delta(v'-v_0),  
\end{equation*}
given an encounter of the susceptible individual with an infectious one (belonging to either $I_1$ or $I_2$) carrying viral load $w$ and for whom $P((i,w) \rightarrow (i,w))=1$, $i\in\{I_1,I_2\}$. 

Infectious 
and recovered individuals cannot change their viral load in binary interactions, but the evolution reflects  physiological processes. Starting from the initial positive value $v=v_0$, the viral load increases until reaching a given peak value and then it decreases towards zero. 
In this framework, the microscopic state $v$  varies as a consequence of an autonomous process (also called \textit{interaction with a fixed background} in the jargon of multi--agent systems \cite{pareschi2013BOOK}). 
Specifically,
given an agent $(I_1,v)$, namely an infectious individual with increasing viral load, we consider a linear--affine expression for the microscopic rule describing the evolution of $v$ into a new viral load $v'$:
\begin{equation}\label{eq:micro_pre.max}
	v'=v+\nu_1(1-v). 
\end{equation}
The latter is a prototype rule describing the fact that the viral load may increase up to a certain threshold normalized to 1 by a factor proportional to $(1-v)$. In particular, $\nu_1\in (0,1)$ is the factor of increase of the viral load.

Similarly, given an agent $(I_2,v)$ 
or $(R,v)$, namely an infectious individual with decreasing viral load or a recovered individual, we consider the following microscopic rule for the evolution of $v$:
\begin{equation}\label{eq:micro_post.max}
	v'=v-\nu_2 v,
\end{equation}
being the parameter $\nu_2\in(0,1)$ the factor of decay of the viral  load. These microscopic processes happen with frequency $\lambda_\gamma>0$, i.e. ${1}/{\lambda_\gamma}$ is the average increase/decay time of the viral load.

We observe here that the introduction of the sub--classes $I_1,\,I_2$  
is needed in order to implement the microscopic rules  \eqref{eq:micro_pre.max}--\eqref{eq:micro_post.max} in a kinetic equation. 
These two rules are deliberately generic and very simple: the only aim is to distinguish individuals based on whether their viral load is increasing or decreasing and to implement two different factors $\nu_1$ and $\nu_2$ accordingly.

We assume that individuals in $I_1$  move to the class $I_2$ with frequency $\lambda_\gamma$ and constant probability $\nu_1$. In turn, individuals in $I_2$ move to the recovered class with frequency $\lambda_\gamma$ and constant probability $\nu_2$. These choices, that trace the same assumptions done in the paper \cite{DMrLnTa22}, allow to derive 
 $\lambda_\gamma \nu_1$ as the rate of transition from $I_1$ to $I_2$, that is also the increase rate of the viral load. Analogously, the rate of recovery from the disease turns to be $\lambda_\gamma \nu_2$, that is  the decay rate of the viral load. Transitions $I_1\rightarrow I_2$ and $I_2\rightarrow R$ are assumed to take place at the same frequency $\lambda_\gamma$ because they are driven by a common cause, namely the progression of the viral load. Hence, the rate of both transitions coincides with the progression rate of the viral load.
Formally, to describe these microscopic mechanisms in terms of transition probabilities, we set
\begin{align*}
	&P_{I_1}(v\rightarrow v')=\delta\left(v'-(v+\nu_1(1-v))\right),\\
	&P_{i}(v\rightarrow v')=\delta\left(v'-(v-\nu_2 v))\right), \quad i\in\{I_2,R\},\\
		&P((I_1,v)\rightarrow (I_2,v'))=\nu_1P_{I_1}(v\rightarrow v'), \\ 
		& P((I_2,v)\rightarrow (R,v'))=\nu_2P_{I_2}(v\rightarrow v').
\end{align*}

\subsection{The kinetic model and the derivation of the macroscopic model}\label{Sec:micro}

In order to give a statistical description of the multi--agent system, whose total mass is conserved in time, we introduce a distribution function for describing the statistical distribution of the individuals characterized by the pair $(x,v) \in \mathcal{X}\times[0,1]$, as
\begin{equation}\label{eq:f.delta.label_switch_net}
	f(t,x, v)={\sum_{i\in \mathcal{X}}} \delta(x-i)  f_{i}(t,v).
\end{equation}
In \eqref{eq:f.delta.label_switch_net},  $\delta(x-i) $  is the Dirac delta distribution centred at $x=i$, and $f_i=f_i(t,v)\geq 0$ is the distribution function of the microscopic state $v$ of the individuals that are in the $i$th compartment at time $t$. Hence, $f_i(t,v)dv$ is the proportion of individuals in the compartment $i$, whose microscopic state lies between $v$ and $v+dv$ at time $t$.

 We  assume that $f(t,x,v)$ is a probability distribution, namely
\begin{equation*}
	\int_{0}^1\int_\mathcal{X} f(t,x,v)dxdv={\sum_{i\in \mathcal{X}}}\int_{0}^1 f_i(t,v)dv=1, \quad \forall\,t\ge0. 
	%\label{eq:f.prob}
\end{equation*}
In general, the $f_i$'s, $i \in \mathcal{X}$, are  not probability density functions because their $v$--integral varies in time due to the fact that individuals move from one compartment to another. 

We denote by
\begin{equation*}
	\rho_i(t)=\int_{0}^1 f_i(t,v)dv
	%\label{eq:rhoi}
\end{equation*}
the density of individuals in the class $i$, thus $0\leq\rho_i(t)\leq 1$ and
$${\sum_{i\in \mathcal{X}}}\rho_i(t)=1, \quad \forall\,t\ge 0. $$
Then, we define the \textit{viral load momentum} of the $i$th compartment as the first moment of $f_i$ for each class $i \in \mathcal{X}$, i.e. 
\begin{equation*}
	n_i(t)=\int_{0}^1 f_i(t,v)vdv.
\end{equation*}
If $\rho_i(t)>0$, then we can also define the \textit{mean viral load} as the ratio
${n_i(t)}/{\rho_i(t)}.
$
Instead, $\rho_i(t)=0$ implies  necessarily  $f_i(t,v)=0$. In this case,  the mean viral load is not defined because the corresponding compartment is empty. %Moreover,  we highlight  that, if the compartment is \textit{almost} empty, then the mean viral load $n_i/\rho_i$, $i\in \mathcal{X}$, might not be fully consistent with the empirical mean viral load resulting from the particle description, because the law of large numbers does not apply.

Starting from the microscopic dynamics illustrated in the previous section, it is possible to formally derive  kinetic equations implementing exactly the microscopic processes (similarly to what done in the paper \cite{DMrLnTa22}). {We report here the \textit{weak} kinetic equations for completeness. Let $\varphi: [0,1]\to\R$ be an arbitrarily chosen test function of an observable quantity depending on the microscopic physical quantity $v$.}
For $i\in\mathcal{X}\setminus\{B,D\}$, namely the classes of living individuals,  we get:
\begin{itemize}
	\item susceptible individuals  ($i=S$)
	\begin{align}
		\begin{aligned}[b]
			\dfrac{d}{dt}\int_0^1\varphi(v)f_S(t,v)dv &= \int_0^1\varphi(v)\left(\lambda_b \dfrac{b}{\rho_B(t)} f_B(t,v)-\lambda_\mu \mu f_S(t,v)\right)dv \\
			&\phantom{=}-\lambda_{\beta} \int_0^1\int_0^1\int_0^1 \varphi(v)\nu_{\beta} (\pr{w})P_S(v) f_S(\pr{v},t)\left(f_{I_1}(t,\pr{w})+f_{I_2}(t,\pr{w})\right) d\pr{v} d\pr{w} dv, 
		\end{aligned}
		\label{eq:boltz.fS1}
	\end{align}
	where  the last term on the r.h.s. accounts for the binary interactions between susceptible individuals and infectious individuals in either $I_1$  ($f_Sf_{I_1}$) or $I_2$ ($f_Sf_{I_2}$),  leading to the transmission of the disease,
	\item infectious individuals with increasing viral  load ($i=I_1$)
	\begin{align}
		\begin{aligned}[b]
			\dfrac{d}{dt}\int_0^1\varphi(v)f_{I_1}(t,v)dv &=-\lambda_\mu \mu\int_0^1\varphi(v) f_{I_1}(t, v) dv \\
			&\phantom{=}+\lambda_{\beta} \int_0^1\int_0^1\int_0^1 \varphi(v) \nu_{\beta} (\pr{w})P_S(v) f_S(\pr{v},t)\left(f_{I_1}(t,\pr{w})+f_{I_2}(t,\pr{w})\right)  d\pr{v} d\pr{w} dv\\
			&\phantom{=}-\lambda_\gamma\nu_1\int_0^1\int_0^1\varphi(v)P_{I_1}(\pr{v} \rightarrow v)f_{I_{1}}(t,\pr{v}) d\pr{v} dv  \\
			&\phantom{=} +\lambda_\gamma\int_0^1\int_0^1 \varphi (v)(P_{I_1}(\pr{v} \rightarrow v)f_{I_{1}}(t,\pr{v})-P_{I_1}(v \rightarrow \pr{v})f_{I_{1}}(t,v)) d\pr{v}  dv,
		\end{aligned}
		\label{eq:boltz.fI_1}
	\end{align}
	\item infectious individuals with decreasing viral  load ($i=I_2$)
	\begin{align}
		\begin{aligned}[b]
			\dfrac{d}{dt}\int_{0}^1\varphi(v)f_{I_2}(t,v)dv &=-\lambda_\mu \mu \int_{0}^1\varphi(v) f_{I_2}(t, v)dv \\
			&\phantom{=}+\lambda_\gamma\nu_1\int_0^1\int_0^1\varphi(v)P_{I_1}(\pr{v} \rightarrow v)f_{I_{1}}(t,\pr{v}) d\pr{v} dv \\
			&\phantom{=}-\lambda_\gamma\nu_2\int_0^1\int_0^1\varphi(v)P_{I_2}(\pr{v} \rightarrow v)f_{I_{2}}(t,\pr{v}) d\pr{v} dv\\
			&\phantom{=}+ \lambda_\gamma\int_0^1\int_0^1 \varphi (v)(P_{I_2}(\pr{v} \rightarrow v)f_{I_{1}}(t,\pr{v})-P_{I_2}(v \rightarrow \pr{v})f_{I_{2}}(t,v)) d\pr{v}  dv,
		\end{aligned}
		\label{eq:boltz.fI_2}
	\end{align}
	\item recovered individuals ($i=R$)
	\begin{align}
		\begin{aligned}[b]
			\dfrac{d}{dt}\int_{0}^1\varphi(v)f_R(t,v)dv &=-\lambda_\mu \mu \int_{0}^1\varphi(v) f_R(t,v)dv  \\
			&\phantom{=}+\lambda_\gamma\nu_2\int_0^1\int_0^1\varphi(v)P_{I_2}(\pr{v} \rightarrow v)f_{I_{2}}(t,\pr{v}) d\pr{v} dv\\
			&\phantom{=} +\lambda_\gamma\int_0^1\int_0^1 \varphi (v)(P_{R}(\pr{v} \rightarrow v)f_{R}(t,\pr{v})-P_{R}(v \rightarrow \pr{v})f_{R}(t,v)) d\pr{v}  dv .
		\end{aligned}
		\label{eq:boltz.fR}
	\end{align}
\end{itemize}
 {As far as the disease transmission rate $\lambda_\beta\nu_\beta(\cdot)$ is concerned, we consider that it is given by
\begin{equation}\label{betaw}
	\lambda_\beta\nu_\beta(w)=\beta w^p, 
\end{equation}
with $\beta$ positive constant and $p\in \lbrace 0,1 \rbrace$.

 The choice $p=0$ corresponds to a constant transmission function, while $p=1$ reflects the experimental evidence that higher is the viral load of an infectious individual higher is his/her ability of transmitting the disease. Of course, formulations different from the linear one could be taken into account. However, since the novelty of this assumption and in absence of exhaustive field data, the linear formulation can be considered as a reasonable approximation at a first step.}

{
In order to obtain the equations for the macroscopic densities and viral load momentum of each compartment, we set $\varphi(v)=v^n$  in \eqref{eq:boltz.fS1}--\eqref{eq:boltz.fR}, with $n=0,1$, respectively. Since we consider interaction rules that are linear in $v$ %[] 
 and we  assume that $\nu_{\beta}(\cdot)$ is a constant or a linear function,  we obtain an exact closed system of macroscopic equations, without the need of other assumptions. This also implies that at the macroscopic level individuals in the same compartment {may have heterogeneous viral loads that can be different from the mean viral load of the compartment}. %Note that, in principle, we could also add a stochastic fluctuation term in the interaction rules  for the evolution of the viral load \eqref{eq:micro_pre.max}--\eqref{eq:micro_post.max}.}

{The ensuing macroscopic model is given by the following system of non--linear ordinary differential equations:
\begin{equation}
	\begin{aligned}
		\dot \rho_S &=   b -\beta\left(\dfrac{n_{I_1}+n_{I_2}}{\rho_{I_1}+\rho_{I_2}}\right)^p\rho_S\left(\rho_{I_1}+\rho_{I_2}\right)  - \mu\rho_S \\
		\dot \rho_{I_1} &=  \beta\left(\dfrac{n_{I_1}+n_{I_2}}{\rho_{I_1}+\rho_{I_2}}\right)^p\rho_S\left(\rho_{I_1}+\rho_{I_2}\right)
		-\lambda_\gamma\nu_1\rho_{I_1} - \mu \rho_{I_1}\\
		\dot \rho_{I_2}&=\lambda_\gamma\nu_1\rho_{I_1}
		-\lambda_\gamma \nu_2\rho_{I_2}- \mu \rho_{I_2}\\
		\dot \rho_R &=\lambda_\gamma\nu_2\rho_{I_2}
		- \mu \rho_R\\
		\dot n_{I_1}&= \beta\left(\dfrac{n_{I_1}+n_{I_2}}{\rho_{I_1}+\rho_{I_2}}\right)^pv_0\rho_S\left(\rho_{I_1}+\rho_{I_2}\right)+\lambda_\gamma\nu_1(1-\nu_1)\rho_{I_1}-\lambda_\gamma\nu_1(2-\nu_1)n_{I_1}
		- \mu n_{I_1} \\
		\dot n_{I_2}&=\lambda_\gamma\nu_1^2\rho_{I_1}+\lambda_\gamma\nu_1(1-\nu_1)n_{I_1}
		-\lambda_\gamma\nu_2(2-\nu_2) n_{I_2}- \mu n_{I_2}  \\
		\dot n_R&=\lambda_\gamma\nu_2 (1-\nu_2) n_{I_2}
		-\lambda_\gamma \nu_2 n_R- \mu n_R,
	\end{aligned}\label{macro_simplified}
\end{equation}
where we have set (with a slight abuse of notation) $$b=\lambda_bb,\,\,\mu=\lambda_\mu\mu,$$
representing the net inflow of susceptibles and the rate of natural death, respectively.
Also, for convenience of notation, in (\ref{macro_simplified}) we have denoted with the upper dot the time derivative and omitted the explicit dependence on time of the state variables.}

{
From system (\ref{macro_simplified}) we note that the equations ruling the evolution of the densities of the compartments  have an SIR structure, but the transmission term may depend on the mean viral load of the infectious population. In the simplest case that $p=0$, i.e. $\nu_\beta(\cdot)$ is a constant function, we retrieve a classical SIR model with \textit{standard incidence} \cite{hethcote}, which reduces to
\begin{equation}
	\begin{aligned}
		\dot \rho_S &=   b - \beta \rho_S(\rho_{I_1}+\rho_{I_2})  - \mu\rho_S %\label{S_classic}
		\\
		\dot \rho_{I_1} &=  \beta \rho_S(\rho_{I_1}+\rho_{I_2})
		-\lambda_\gamma\nu_1\rho_{I_1} - \mu \rho_{I_1}%\label{I1_classic}
		\\
		\dot \rho_{I_2}&=\lambda_\gamma\nu_1\rho_{I_1}
		-\lambda_\gamma \nu_2\rho_{I_2}- \mu \rho_{I_2}%\label{I2_classic}
		,
	\end{aligned}
\label{macroSIR}
\end{equation}
by noting  that the differential equations for $\rho_R$, $n_{I_1}$, $n_{I_2}$ and $n_R$ are independent of the other ones.
In such a case, the analysis of the model turns to be trivial  being the equations for the densities independent of the viral load momentum.}

{
In the present work, we take a step forward by assuming that $\nu_\beta(\cdot)$ is a linear increasing function, i.e. by choosing $p=1$ in the system \eqref{macro_simplified}.}
{With this choice,  the model to be studied eventually reduces to} 
\begin{subequations}
	\begin{align}
		\dot \rho_S &=   b - \beta \rho_S(n_{I_1}+n_{I_2})  - \mu\rho_S \label{S'}\\
		\dot \rho_{I_1} &=  \beta \rho_S(n_{I_1}+n_{I_2})
		-\lambda_\gamma\nu_1\rho_{I_1} - \mu \rho_{I_1}\label{I1'}\\
		\dot \rho_{I_2}&=\lambda_\gamma\nu_1\rho_{I_1}
		-\lambda_\gamma \nu_2\rho_{I_2}- \mu \rho_{I_2}\label{I2'}\\
		\dot n_{I_1}&= \beta  v_0\rho_S(n_{I_1}+n_{I_2})
		+\lambda_\gamma\nu_1(1-\nu_1)\rho_{I_1}-\lambda_\gamma\nu_1(2-\nu_1)n_{I_1}- \mu n_{I_1} \label{n1'}\\
		\dot n_{I_2}&=\lambda_\gamma\nu_1^2\rho_{I_1}+\lambda_\gamma\nu_1(1-\nu_1)n_{I_1}
		-\lambda_\gamma\nu_2(2-\nu_2) n_{I_2}- \mu n_{I_2}\label{n2'},
	\end{align}\label{macro2}
\end{subequations}
by noting that the differential equations for $\rho_{R}$ and  
$n_{R}$ are independent of the other ones.
 
To models (\ref{macroSIR})--(\ref{macro2}) we associate the following generic initial conditions
\begin{equation}\label{CI}
	\rho_S(0)=\rho_{S,0}>0,\,\,\rho_{i}(0)=\rho_{i,0}\geq 0,\,\,n_{i}(0)=n_{i,0}\geq 0,\quad i\in\{I_1,I_2\}.
\end{equation}
Equilibria and stability properties of  model (\ref{macro2})   are investigated in the following section.
\begin{remark}
	If field data concerning a specific disease showed evidence that the probability of disease  transmission non--linearly depends on the viral load, one could implement a non--linear transmission function $\nu_\beta(\cdot)$. In such a case, it could not be possible to obtain an exact closed system of {macroscopic} equations, but other closure assumptions could be required. For example, in the paper \cite{DMrLnTa22} a monokinetic closure is used, implying that in the derivation of the macroscopic model all the individuals of a given compartment are assumed to have as viral load  the mean value of that compartment.
\end{remark}

\section{Qualitative analysis}\label{sec:4}
Let us start by ensuring that the model (\ref{macro2}) is mathematically and epidemiologically well posed.
It is straightforward to verify that the region
\begin{equation*}%\label{regionD}
	\mathcal{D}=\left\{\left(\rho_S,\rho_{I_1},\rho_{I_2},n_{I_1},n_{I_2}\right)\in [0,1]^5 \,\Big|\, 0<\rho_S+\rho_{I_1}+\rho_{I_2}\leq \dfrac{b}{ \mu},\,n_{I_1}\leq\rho_{I_1},\,n_{I_2}\leq\rho_{I_2}\right\}
\end{equation*}
with initial conditions in (\ref{CI}) is positively invariant for model (\ref{macro2}), namely any solution of system (\ref{macro2}) starting in $\mathcal{D}$ remains in $\mathcal{D}$ for all $t\geq 0$.

\subsection{The disease--free equilibrium and the basic reproduction number}

The model  (\ref{macro2}) has a unique disease--free equilibrium (DFE),  given by
\begin{equation}
	DFE=\left(\dfrac{b}{\mu}, 0,0,0,0\right).\label{DFE}
\end{equation}
It is obtained
by setting the r.h.s. of equations  (\ref{macro2})  to zero and considering the case $\rho_{I_1}=\rho_{I_2}=0$.
\begin{theorem}\label{ProplocalDFE}
	The DFE of model  (\ref{macro2}) is locally asymptotically  stable (LAS) if $\mathcal{R}_0<1$, where
	\begin{equation}\label{R0}
\mathcal{R}_0=\beta \dfrac{ b}{ \mu}\dfrac{\lambda_\gamma\nu_1(\lambda_\gamma\nu_1(1-\nu_2)^2+\lambda_\gamma\nu_2(2-\nu_2)+\mu)+v_0(\lambda_\gamma\nu_1   + \mu)(\lambda_\gamma\nu_1(1-\nu_1)   +\lambda_\gamma\nu_2(2-\nu_2)   + \mu)}{(\lambda_\gamma\nu_1   + \mu)(\lambda_\gamma\nu_1(2-\nu_1)   + \mu)(\lambda_\gamma\nu_2(2-\nu_2)   + \mu)}.
	\end{equation}
	 Otherwise, if $\mathcal{R}_0>1$, then it is unstable.
\end{theorem}
\begin{proof}
The Jacobian matrix of system  (\ref{macro2}) 
evaluated at the DFE (\ref{DFE}) 
reads
\begin{equation*}
	J(DFE)=\left(\begin{array}{ccccc}	  - \mu&		0   &		0 & -{\beta} \dfrac{ b}{ \mu}   &		-{\beta} \dfrac{ b}{ \mu}  \\
		0 & -\lambda_\gamma\nu_1   - \mu & 0  &{\beta} \dfrac{ b}{ \mu}&{\beta} \dfrac{ b}{ \mu}\\
		0&\lambda_\gamma\nu_1    &-\lambda_\gamma\nu_2-\mu&0&0\\
		0&	\lambda_\gamma \nu_1(1-\nu_1) & 	0&{\beta}v_0\dfrac{ b}{ \mu}-\lambda_\gamma\nu_1(2-\nu_1) -\mu&{\beta}v_0\dfrac{ b}{ \mu}\\
		0&\lambda_\gamma \nu_1^2&0&\lambda_\gamma \nu_1(1-\nu_1)  &-\lambda_\gamma\nu_2(2-\nu_2) -\mu
	\end{array}\right).
\end{equation*}
One can immediately get the eigenvalues $l_1= - \mu<0$, $l_2=-\lambda_\gamma\nu_2- \mu$, while the other three are determined by the submatrix
$$\bar J=	\left(\begin{array}{ccccc}	
	-\lambda_\gamma\nu_1   - \mu   &{\beta} \dfrac{ b}{ \mu}&{\beta} \dfrac{ b}{ \mu}\\
	\lambda_\gamma \nu_1(1-\nu_1) & 	{\beta}v_0\dfrac{ b}{ \mu}-\lambda_\gamma\nu_1(2-\nu_1) -\mu&{\beta}v_0\dfrac{ b}{ \mu}\\
	\lambda_\gamma \nu_1^2&\lambda_\gamma \nu_1(1-\nu_1)  &-\lambda_\gamma\nu_2(2-\nu_2) -\mu
\end{array}\right). $$
The characteristic polynomial of $\bar J$ reads
$$p(l)=l^3+a_1l^2+a_2l+a_3,$$
where
\begin{align*}
	a_1&=-\text{Tr}(\bar J)=-\bar J_{11}-\bar J_{22}-\bar J_{33}\\
	a_2&=\dfrac{1}{2}\left(\text{Tr}^2(\bar J)-\text{Tr}(\bar J^2)\right)=
	\left(\bar J_{22}-\beta v_0\dfrac{b}{\mu}\right) \bar J_{33}(1-\mathcal{R}_0)+\bar J_{11}\left(\bar J_{22}+\bar J_{33}\right)-\beta\dfrac{b}{\mu}\dfrac{\lambda_\gamma^2\nu_1\nu_2(1-\nu_1)(2-\nu_2)}{\bar J_{11}}%\nonumber
	\\
	a_3&=-\text{Det}(\bar J)=-\bar J_{11}\left(\bar J_{22}-\beta v_0\dfrac{b}{\mu}\right)\bar J_{33}(1-\mathcal{R}_0),
\end{align*}
 with $\mathcal{R}_0$ given in (\ref{R0}).
 
In particular, sgn$(a_3)$=sgn$(1-\mathcal{R}_0)$.
 Also, $\mathcal{R}_0<1$ implies that $\bar J_{22}<0$, yielding $a_1>0$ and $a_1a_2-a_3>0$.

From the Routh--Hurwitz criterion it follows that, if $\mathcal{R}_0<1$,
then the DFE is LAS. Otherwise, if $\mathcal{R}_0>1$,  then it is unstable.\end{proof}

The threshold quantity $\mathcal{R}_0$ is the so--called \textit{basic reproduction number} for model (\ref{macro2}), a frequently used indicator for measuring the potential spread of an infectious disease in a community. Epidemiologically, it represents the average number of secondary cases produced by one primary infection over the course of the infectious period in a fully susceptible population. 

The expression of $\mathcal{R}_0$ for model (\ref{macro2}) turns out to be much more complex than that for the  epidemic model (\ref{macroSIR}) which assumes a constant disease transmission rate (we investigate more in details this point in the subsection \ref{Sec:R0num}). Note that the $\mathcal{R}_0$ in (\ref{R0}) depends also on $v_0$, the initial viral load of infectious individuals, a parameter that is not present in the differential equations for the densities of the compartments, namely (\ref{S'})--(\ref{I1'})--(\ref{I2'}).
\begin{remark}
	It would be interesting to investigate how  the expression of
	the basic reproduction number $\mathcal{R}_0$ varies by modifying some assumptions of model (\ref{macro2}). For instance, one can consider the case that individuals in one between the classes $I_1$ and $I_2$ are infected but not infectious. Specifically, one can assume that
	\begin{itemize}
		\item individuals in $I_1$ are not infectious because, for the specific disease, the period of viral load increase can be approximated to the period of latency of the infection. In such a case, the $I_1$'s play the role of the \textit{exposed} individuals $E$ in an SEIR model. This leads to the disappearance of the term $\beta\rho_Sn_{I_1}$ [resp. $\beta v_0\rho_Sn_{I_1}$] in the equation (\ref{I1'}) [resp. (\ref{n1'})]. The basic reproduction number proves to be 
		\begin{equation*}
			\mathcal{R}_0={\beta} \dfrac{ b}{\mu }\dfrac{\lambda_\gamma\nu_1^2(\lambda_\gamma+\mu)+v_0(\lambda_\gamma\nu_1   + \mu)\lambda_\gamma\nu_1(1-\nu_1)   }{(\lambda_\gamma\nu_1   + \mu)(\lambda_\gamma\nu_1(2-\nu_1)   + \mu)(\lambda_\gamma\nu_2(2-\nu_2)   + \mu)};
		\end{equation*}
		\item individuals in $I_2$ are not infectious because they are isolated from the community and receive treatment to decrease the viral load. This leads to the disappearance of the term $\beta\rho_Sn_{I_2}$ [resp. $\beta v_0\rho_Sn_{I_2}$] in the equation (\ref{I1'}) [resp. (\ref{n1'})]. In such a case, the basic reproduction number proves to be 
		\begin{equation*}
			\mathcal{R}_0={\beta} \dfrac{ b}{\mu }\dfrac{\lambda_\gamma\nu_1(1-\nu_1)+v_0(\lambda_\gamma\nu_1   + \mu)   }{(\lambda_\gamma\nu_1   + \mu)(\lambda_\gamma\nu_1(2-\nu_1)   + \mu)}.	
		\end{equation*}
	\end{itemize}
\end{remark}

\subsection{The endemic equilibrium}
Let us denote by
\begin{equation}
EE=\left(\rho_S^E,\rho_{I_1}^E,\rho_{I_2}^E,n_{I_1}^E,n_{I_2}^E\right)\label{EE}
\end{equation}
the generic endemic equilibrium of model (\ref{macro2}), obtained
by setting the r.h.s. of equations  (\ref{macro2})  to zero and considering the case $\rho_{I_1}+\rho_{I_2}>0$. 
Note that if it were $\rho_{I_1}^E=0$ [resp. $\rho_{I_2}^E=0$], from (\ref{I2'}) it would follow that  $\rho_{I_2}^E=0$  [resp. $\rho_{I_1}^E=0$].  Hence, it must be $\rho_{I_1}^E,\,\rho_{I_2}^E>0$.

More precisely, by rearranging equations (\ref{S'})--(\ref{I1'})--(\ref{I2'})--(\ref{n1'}), one obtains
\begin{equation}
\begin{aligned}
	\rho_S^E&=\dfrac{ b- (\lambda _{\gamma }\nu _1 +\mu ) \rho_{I_1}^E}{\mu }\\
	\rho_{I_1}^E&=n_{I_1}^E\dfrac{\lambda _{\gamma }\nu _1(2-\nu_1) +\mu  }{ \lambda _{\gamma }\nu _1(1-\nu _1)+v_0(\lambda _{\gamma }\nu _1 + \mu  )}\\
	\rho_{I_2}^E&=\rho_{I_1}^E\dfrac{\lambda _\gamma\nu_1}{\lambda_\gamma\nu_2+\mu}\\
	n_{I_2}^E&=\dfrac{b -{\beta } \rho _S^En_{I_1}^E -\mu   \rho _S^E}{{\beta } \rho _S^E}.
\end{aligned}\label{EEcomp}
\end{equation}
Substituting the expressions (\ref{EEcomp}) into (\ref{n2'}), one gets  $n_{I_1}^E$  as a positive root of the equation
\begin{equation*}
 \lambda_\gamma\nu_1(n_{I_1}^E+\nu_1 (\rho_{I_1}^E-n_{I_1}^E))-\lambda_\gamma\nu_2(2-\nu_2) n_{I_2}^E- \mu n_{I_2}^E=0 ,	
\end{equation*}	
that is
\begin{equation}\label{n1EE}
	n_{I_1}^E=b\dfrac{\lambda _{\gamma }\nu _1(1-\nu_1) +v_0(\lambda_\gamma\nu_1+\mu)}{(\lambda_\gamma\nu_1+\mu)(\lambda_\gamma\nu_1(2-\nu_1)+\mu)}\left(1-\dfrac{1}{\mathcal{R}_0}\right).
\end{equation}
Then, one can make explicit also the other components of $EE$:
\begin{equation}
	\begin{aligned}
		\rho_S^E&=\dfrac{ b}{\mu }\dfrac{ 1}{\mathcal{R}_0 }\\
		\rho_{I_1}^E&=\dfrac{b}{ \lambda _{\gamma }\nu _1 + \mu  }\left(1-\dfrac{1}{\mathcal{R}_0}\right)\\
		\rho_{I_2}^E&=b\dfrac{\lambda _{\gamma }\nu _1}{ (\lambda _{\gamma }\nu _1 + \mu )(\lambda _{\gamma }\nu _2 + \mu ) }\left(1-\dfrac{1}{\mathcal{R}_0}\right)\\
		n_{I_2}^E&=\dfrac{\mu   }{\beta }\dfrac{\lambda _{\gamma }\nu _1^2(\lambda_\gamma+\mu) +v_0\lambda_\gamma\nu_1(1-\nu_1)(\lambda_\gamma\nu_1+\mu)}{\lambda_\gamma\nu_1(\lambda_\gamma\nu_1(1-\nu_2)^2+\lambda_\gamma\nu_2(2-\nu_2)+\mu)+v_0(\lambda_\gamma\nu_1   + \mu)(\lambda_\gamma\nu_1(1-\nu_1)   +\lambda_\gamma\nu_2(2-\nu_2)   + \mu)}\left({\mathcal{R}_0}-1\right).
	\end{aligned}\label{EEexplicit}
\end{equation}
For the equilibrium to exist in $\mathcal{D}$ all its components must be positive. Hence, the following result can be stated.
\begin{theorem}
	\label{EEexistence}
	If $\mathcal{R}_0<1$, then the model (\ref{macro2}) has no endemic equilibria. Otherwise, if $\mathcal{R}_0>1$,  then the model (\ref{macro2}) has   an un unique endemic equilibrium (\ref{EE}) whose components are given in (\ref{n1EE})--(\ref{EEexplicit}). 
\end{theorem}
Due to the complexity of the Jacobian matrix of system (\ref{macro2}) evaluated at $EE$, we renounce to study the local stability of the endemic equilibrium. However, we  make use of bifurcation analysis and show that a unique
branch corresponding to the unique endemic equilibrium emerges from the criticality, namely at DFE and $\mathcal{R}_0=1$. The emerging $EE$ is LAS in the neighbouring of $\mathcal{R}_0=1$ for $\mathcal{R}_0>1$.
	
	\subsection{Central manifold analysis}
	
	To derive a sufficient condition for the occurrence of a transcritical  bifurcation at $\mathcal{R}_0=1$, we can use
	a bifurcation theory approach. We adopt the approach developed in
	\cite{dushoff1998,vandendriessche2002}, which is based
	on the general center manifold theory \cite{guckenheimer1983}. In
	short, it establishes that the normal form representing the dynamics
	of the system on the central manifold is, for $u$ sufficiently small, given by
	$$
	\dot{u}=A{u}^{2}+B\beta{u},
	$$
	where
	\begin{equation}
		A=\dfrac{\mathbf{z}}{2}\cdot D_\mathbf{{xx}}\mathbf{F}(DFE,\overline{\beta})\mathbf{w}^{2}\equiv\dfrac{1}{2}{\sum_{k,i,j=1}^{5} z_{k}w_{i}w_{j}\dfrac{\partial^{2}F_{k}(DFE,\overline{\beta})}{\partial x_{i}\partial x_{j}}} \label{eq:a}
	\end{equation}
	and
	\begin{equation}
		B=\mathbf{z}\cdot D_{\mathbf{x}\beta}\mathbf{F}(DFE,\overline{\beta})\mathbf{w}\equiv{\sum^{5}_{k,i=1}}z_{k}w_{i}\dfrac{\partial^{2}F_{k}(DFE,\overline{\beta})}{\partial x_{i}\partial\beta}.\label{eq:b}
	\end{equation}
	Note that in (\ref{eq:a}) and (\ref{eq:b}) the transmission rate $\beta$ has been chosen
	as bifurcation parameter, $\overline{\beta}$ is the critical value of $\beta$, $\mathbf{x}=\left(\rho_S,\rho_{I_1},\rho_{I_2},n_{I_1},n_{I_2}\right)$ is the state variables vector,
	$\mathbf{F}$ is the r.h.s. of system (\ref{macro2}),
	and $\mathbf{z}$ and $\mathbf{w}$
	denote, respectively, the left and right eigenvectors corresponding
	to the null eigenvalue of the Jacobian matrix evaluated at criticality
	(i.e. at DFE and $\beta=\overline{\beta}$).
	
	Observe that $\mathcal{R}_0=1$ is equivalent to
	\[
	{\beta}=\overline{\beta}= \dfrac{ \mu}{ b}\dfrac{(\lambda_\gamma\nu_1   + \mu)(\lambda_\gamma\nu_1(2-\nu_1)   + \mu)(\lambda_\gamma\nu_2(2-\nu_2)   + \mu) }{\lambda_\gamma\nu_1(\lambda_\gamma\nu_1(1-\nu_2)^2+\lambda_\gamma\nu_2(2-\nu_2)+\mu)+v_0(\lambda_\gamma\nu_1   + \mu)(\lambda_\gamma\nu_1(1-\nu_1)   +\lambda_\gamma\nu_2(2-\nu_2)   + \mu)},
	\]
	so that the disease--free equilibrium is  LAS if $\beta<\overline{\beta}$,
	and it is unstable when $\beta>\overline{\beta}$. 
	
	The direction of the bifurcation occurring at $\beta=\overline{\beta}$ can
	be derived from the sign of coefficients (\ref{eq:a}) and (\ref{eq:b}).
	More precisely, if $A>0$ [resp. $A<0$] and $B>0$, then at $\beta=\overline{\beta}$ there
	is a backward [resp. forward] bifurcation.
	
	For our model, we prove the following theorem.
	\begin{theorem}
		System  (\ref{macro2}) exhibits a  forward bifurcation at DFE
		and $\mathcal{R}_0=1$.
		\begin{proof}
			From the proof of Theorem \ref{ProplocalDFE}, one can  verify that, when $\beta=\overline{\beta}$  (or, equivalently,
			when $\mathcal{R}_0=1$), the Jacobian matrix $J(DFE)$ admits a simple zero eigenvalue and the other eigenvalues
			have negative real part. Hence, the
			DFE is a non--hyperbolic equilibrium.
			
			It can be easily checked that a left and a right eigenvector associated
			with the zero eigenvalue so that $\mathbf{z\cdot}\mathbf{w}=1$ are
			\begin{gather*}
				\mathbf{z}=\left(0,\lambda_{\gamma }\nu_1\dfrac{\lambda_{\gamma}\nu _2(1-\nu_1)(2-\nu_2)+\lambda_{\gamma}\nu_1+\mu}{(\lambda_{\gamma}\nu_1+\mu)(\lambda_{\gamma}\nu_1(1-\nu_1)+\lambda_{\gamma}\nu_2(2-\nu_2)+\mu)}z_4,0,z_4,\dfrac{\lambda _{\gamma} \nu _1\left(2-\nu _1\right) +\mu}{\lambda _{\gamma} \nu _1\left(1-\nu _1\right)+\lambda _{\gamma} \nu _2\left(2-\nu _2\right) +\mu}z_4\right)\\
				\mathbf{w}=\left(-\dfrac{ \lambda _{\gamma }\nu _1 + \mu }{ \mu },1,\dfrac{  \lambda _{\gamma }\nu _1}{ \lambda _{\gamma }\nu _2+\mu},\dfrac{v_0(\lambda _{\gamma }\nu _1+\lambda _{\mu }\mu) +\lambda _{\gamma }\nu _1(1-\nu _1) }{ \lambda _{\gamma} \nu _1\left(2-\nu _1\right) +\lambda _{\mu }\mu  },w_5\right)^{T},
			\end{gather*}
			with
		\begin{equation*}
			z_4=
			\dfrac{K}{\lambda_{\gamma }\nu_1\left[\lambda_{\gamma}\nu _2(1-\nu_1)(2-\nu_2)+\lambda_{\gamma}\nu_1+\mu\right]+Kw_4+(\lambda_{\gamma}\nu_1+\mu)(\lambda_{\gamma}\nu_1(2-\nu_1)+\mu)w_5}\end{equation*}
			$$w_5=
			\lambda _{\gamma }\nu _1\dfrac{ \nu _1 \left(\lambda _{\gamma }+\lambda _{\mu }\mu  \right)+ v_0\left(1-\nu _1\right) \left(\lambda _{\gamma }\nu _1 +\lambda _{\mu }\mu  \right)}{\left( \lambda _{\gamma }\nu _1 \left(2-\nu _1\right)+ \lambda _{\mu }\mu \right) \left(\lambda _{\gamma }\nu _2 \left(2-\nu _2\right) + \lambda _{\mu }\mu \right)}, $$
			and
			$$K=(\lambda_{\gamma}\nu_1+\mu)(\lambda_{\gamma}\nu_1(1-\nu_1)+\lambda_{\gamma}\nu_2(2-\nu_2)+\mu).$$
			The coefficients $A$ and $B$ may be now explicitly computed. Considering
			only the non--zero components of the eigenvectors and computing the
			corresponding second derivative of $\mathbf{F}$, it follows that
			\begin{align*}
				A&=z_2w_1w_4\dfrac{\partial^{2}F_{2}(DFE,\overline{\beta})}{\partial \rho_S\partial n_{I_1}}+z_2w_1w_5\dfrac{\partial^{2}F_{2}(DFE,\overline{\beta})}{\partial \rho_S\partial n_{I_2}}+z_4w_1w_4\dfrac{\partial^{2}F_{4}(DFE,\overline{\beta})}{\partial \rho_S\partial n_{I_1}}+z_4w_1w_5\dfrac{\partial^{2}F_{4}(DFE,\overline{\beta})}{\partial \rho_S\partial n_{I_2}}\\&=\overline{\beta}\left(z_2+v_0z_4\right)w_1(w_4+w_5)
			\end{align*}
			and
			\begin{align*}
			B&=z_{2}w_4\dfrac{\partial^{2}F_{2}(DFE,\overline{\beta})}{\partial n_{I_1}\partial\beta}+z_{2}w_5\dfrac{\partial^{2}F_{2}(DFE,\overline{\beta})}{\partial n_{I_2}\partial\beta}+z_{4}w_4\dfrac{\partial^{2}F_{4}(DFE,\overline{\beta})}{\partial n_{I_1}\partial\beta}+z_{4}w_5\dfrac{\partial^{2}F_{4}(DFE,\overline{\beta})}{\partial n_{I_2}\partial\beta}\\&= \dfrac{ b}{\mu}(z_2+v_0z_4)(w_4+w_5),
				\end{align*}
			where $z_2,\,z_4,\,w_4,\,w_5>0$ and $w_1<0$. Then, $A<0<B$. Namely, when $\beta-\overline{\beta}$ changes from negative to positive, the
			DFE changes its stability from locally asymptotically stable to unstable; correspondingly, an endemic and locally asymptotically stable equilibrium emerges. This completes the proof.
	\end{proof} \end{theorem}

\section{Numerical simulations}\label{sec:5}
In this section, we numerically investigate how the  viral load of the infectious individuals may affect the disease transmission among the population.
At this aim, we compare the basic reproduction number and the  numerical solutions of the macroscopic model (\ref{macro_simplified}) in the case of viral load--dependent rate of disease transmission ($p=1$) with those in the {classical} case of constant rate of disease transmission ($p=0$). 

Numerical simulations are performed in \textsc{Matlab}\textsuperscript{\textregistered} \cite{ma}.  We implement  the 4th order  Runge--Kutta method with constant step size for integrating the system (\ref{macro_simplified}). Platform--integrated functions are used for getting the plots. 

\subsection{Parametrization}\label{Sec:par}
\begin{table}[t!]
	\centering\begin{tabular}{|@{}c|c|c@{}|}
		\hline
		Parameter& Description &Baseline value \\
		\hline
			$b$ & Net inflow of susceptibles &$2.11\cdot 10^{-5}$ days$^{-1}$\\	
			$\mu$&  Rate of natural death &$3.18\cdot 10^{-5}$ days$^{-1}$\\	
			$\beta^c$& Constant transmission rate &See text\\	
			$M$& Factor of transmission normalization  &See text\\	
			$\beta^v$& Viral load--dependent transmission factor &$\beta^c/M$\\	
			$v_0$&Initial viral load of infectious individuals &0.01\\	
			$\lambda_\gamma$& Frequency of viral load evolution&0.5 days$^{-1}$\\	
			$\nu_1$&Factor of increase of the viral load&0.4\\	
			$\nu_2$&Factor of decay of the viral load&0.2\\	\hline
	\end{tabular}\caption{List of  model parameters with corresponding description and baseline value.}\label{TabPar}
\end{table}
Since our investigations are purely qualitative, demographic and epidemiological parameters values do not address a specific infectious disease and/or spatial area. They refer to a generic epidemic course following an SIR--like dynamics.

We are considering  a  model with demography and constant net inflow of susceptibles $b$. Since travel restrictions are usually implemented during epidemics, we assume that $b$ accounts only for
new births (which can be assumed to be approximately constant due to the short time span of our analyses). Therefore, the net inflow of susceptibles is given by
	\begin{equation*}
		%\label{netinflow}
		 b=b_r\dfrac{\bar{N}}{N_{tot}},
	\end{equation*}
	where $b_r$ is the birth rate, $\bar N$ denotes the total resident population at the initial time, and 
	$N_{tot}$ is  the total (constant) system size. Note that $N_{tot}$ accounts for individuals belonging to all model compartments $\mathcal{X}$ (including $B$, $D$), whereas $\bar N$ refers only to living individuals. 
	
	We assume an initial population of $\bar N=10^6$ individuals, representing, for example, the  inhabitants of a European metropolis. 
The most recent data by European Statistics  refer to  2020 and provide an average crude birth rate $b_r=9.1/1,000$ years$^{-1}$ \cite{euros1} and an average crude death rate $\mu=11.6/1,000$ years$^{-1}$ \cite{euros2}. 
The total system size $N_{tot}$ is set to $N_{tot}=\bar N/(1-bt_{max})$, in such a way $N_{tot}=\bar N +bt_{max} N_{tot}$ is given by the sum of the initial population, $\bar N$, and the total inflow of individuals during the time interval $[0,t_{max}]$, $bt_{max} N_{tot}$. The time $t_{max}$  is set to $t_{max}=20$ years, that is much larger than the terminal time of our numerical simulations, so ensuring that the compartment $B$ remains not empty.

As far as the disease transmission rate $\lambda_\beta\nu_\beta(\cdot)$ is concerned,  we numerically compare the characteristics of the disease dynamics in the case that $\nu_\beta(\cdot)$ depends on the individual viral load ($p=1$ in (\ref{betaw})) w.r.t the classical case that $\nu_\beta(\cdot)$ is constant ($p=0$ in (\ref{betaw})). Namely, we
consider the following simulation scenarios:
\begin{enumerate}
	\item[\textbf{S$^v$}]  \label{S1} viral load--dependent transmission rate, as studied here: $\lambda_\beta\nu_\beta(w)=\beta^v w$ (i.e., model (\ref{macro2}) with $\beta^v$ in place of $\beta$);
	\item[\textbf{S$^c$}] \label{S2} constant transmission rate, as in {classical} epidemic models: $\lambda_\beta\nu_\beta(w)=\beta^c$ (i.e., model (\ref{macroSIR}) with $\beta^c$ in place of $\beta$).
\end{enumerate}
In order to make the two scenarios properly comparable, we make the following considerations. In the case \textbf{S$^c$}, the quantity $\beta^c$ represents the rate at which infectious individuals transmit the disease in the unit of time. In the case \textbf{S$^v$}, in the microscopic model the same rate is given by $\beta^v$ multiplied by the  microscopic viral load $w$ of the infectious individual ${I_j}$, $j\in\{1,2\}$; whereas, in the macroscopic model (\ref{macro_simplified}) this rate is given by $\beta^v$ multiplied by the mean viral load of the total infectious population: $(n_{I_1}+n_{I_2})/(\rho_{I_1}+\rho_{I_2})$. Thus, we assume that the value of $\beta^v$ in  scenario \textbf{S$^v$}  is given by the value $\beta^c$ adopted in  scenario \textbf{S$^c$}  rescaled by a normalization factor $M\in(0,1)$:
\begin{equation}\label{betavc}
	\beta^v=\dfrac{\beta^c}{M},
\end{equation}
where $M$ represents an \textit{average} quantity for  $(n_{I_1}+n_{I_2})/(\rho_{I_1}+\rho_{I_2})$. It follows that
$$\beta^v>\beta^c.$$
 For the other epidemiological parameters we take the following baseline values from the paper \cite{DMrLnTa22}:
 \begin{equation*}
  \lambda_\gamma=1/2 \text{ days}^{-1},\,\, \nu_1=1/(5\lambda_\gamma),\,\, \nu_2=\nu_1/2,\,\,v_0=0.01.
\end{equation*}
In particular, the product $\lambda_\gamma\nu_1$ can be interpreted as the inverse of the average time from exposure to viral load peak, whilst $\lambda_\gamma\nu_2$ as the inverse of the average time from viral load peak to recovery.  

All the parameters of the model as well as their baseline values are reported in Table \ref{TabPar}.

\subsection{Impact of viral load on the reproduction number}\label{Sec:R0num}
In this subsection, we investigate the impact of different modelling assumptions about the disease transmission rate on the expression and value of the reproduction number of model (\ref{macro_simplified}).

At this aim, let us denote by 
\begin{equation}\label{R0v}
	\mathcal{R}_0^v=\beta^v \dfrac{ b}{ \mu}\dfrac{\lambda_\gamma\nu_1(\lambda_\gamma\nu_1(1-\nu_2)^2+\lambda_\gamma\nu_2(2-\nu_2)+\mu)+v_0(\lambda_\gamma\nu_1   + \mu)(\lambda_\gamma\nu_1(1-\nu_1)   +\lambda_\gamma\nu_2(2-\nu_2)   + \mu)}{(\lambda_\gamma\nu_1   + \mu)(\lambda_\gamma\nu_1(2-\nu_1)   + \mu)(\lambda_\gamma\nu_2(2-\nu_2)   + \mu)}
\end{equation}
the basic reproduction number of model (\ref{macro_simplified}) in the case of viral load--dependent transmission rate, \textbf{S$^v$}, that is (\ref{R0}) with $\beta^v$ in place of $\beta$.  It is straightforward to verify that in the case of constant disease transmission rate,  \textbf{S$^c$}, the basic reproduction number of model (\ref{macro_simplified})  reads
\begin{equation}\label{R0c}
	\mathcal{R}_0^c=\beta^c \dfrac{ b}{ \mu}\dfrac{\lambda_\gamma\nu_1+\lambda_\gamma\nu_2   + \mu }{(\lambda_\gamma\nu_1   + \mu)(\lambda_\gamma\nu_2   + \mu)},
\end{equation}
(see also the paper \cite{DMrLnTa22}).

\begin{remark}\label{rmklimit}
It is difficult to determine \textit{a priori} the relationship between  $\mathcal{R}_0^v$ and $\mathcal{R}_0^c$ for a given set of parameters. 
Nonetheless,  some considerations can be made in the limit cases:
\begin{itemize}
	\item[i)] $\lambda_\gamma\gg 1$, namely the viral load of an infected individual evolves very slowly. \\
	Then, by considering the numerator and the denominator of $\mathcal{R}_0^v$ and $\mathcal{R}_0^c$ as polynomials in $\lambda_\gamma$ and disregarding the lower order terms, one can approximate
	 \begin{equation*}
	 	\mathcal{R}_0^v\approx\beta^v \dfrac{ b}{ \mu}\dfrac{\nu_1(1-\nu_2)^2+\nu_2(2-\nu_2)+v_0(\nu_1(1-\nu_1)   +\nu_2(2-\nu_2))}{\lambda_\gamma\nu_1\nu_2(2-\nu_1)(2-\nu_2)},\,\,\mathcal{R}_0^c\approx\beta^c \dfrac{ b}{ \mu}\dfrac{\nu_1+\nu_2}{\lambda_\gamma   \nu_1\nu_2}.
	 \end{equation*}
 Interestingly, the ratio $\mathcal{R}_0^v/\mathcal{R}_0^c$ turns to be independent of $\lambda_\gamma$.
		\item[ii)] $\nu_1\to 0$, namely all the infectious individuals have constant viral load $v_0$ and do not recover from the disease. \\
	Then, $\mathcal{R}_0^v$ and $\mathcal{R}_0^c$ coincide. Indeed,
	$$\mathcal{R}_0^v=\beta^v \dfrac{ bv_0}{ \mu^2}=\beta^c \dfrac{ b}{ \mu^2}=\mathcal{R}_0^c,$$
	being $M=v_0$ in (\ref{betavc}).
	\item[iii)]  $\nu_1, \nu_2\to 1$, namely in the two subsequent evolution steps after the infection, the viral load of the infected individuals reaches the maximum value 1 and then vanishes, respectively.\\
	Then, the reproduction numbers $\mathcal{R}_0^v$ and $\mathcal{R}_0^c$ read
	\begin{equation*}
		\mathcal{R}_0^v=\beta^v \dfrac{ b}{ \mu}\dfrac{\lambda_\gamma+v_0(\lambda_\gamma   + \mu)}{(\lambda_\gamma   + \mu)^2},\,\,\mathcal{R}_0^c=\beta^c \dfrac{ b}{ \mu}\dfrac{2\lambda_\gamma+ \mu }{(\lambda_\gamma   + \mu)^2},
	\end{equation*}
implying that
\begin{equation*}
	\text{sgn}\left(\mathcal{R}^v_0-\mathcal{R}^c_0\right)=\text{sgn}\left((1-M)\lambda_{\gamma}-(M-v_0)(\lambda_\gamma   + \mu)\right).
\end{equation*}
In such a case, if we further assume that $\lambda_\gamma\gg 1$ (see point (i)), then the ratio $\mathcal{R}_0^v/\mathcal{R}_0^c$ reduces to
\begin{equation*}
	\dfrac{\mathcal{R}_0^v}{\mathcal{R}_0^c}=\dfrac{1+v_0}{2M}.
\end{equation*}
\end{itemize}
\end{remark}
\begin{figure}[t]\centering
	\includegraphics[scale=0.9]{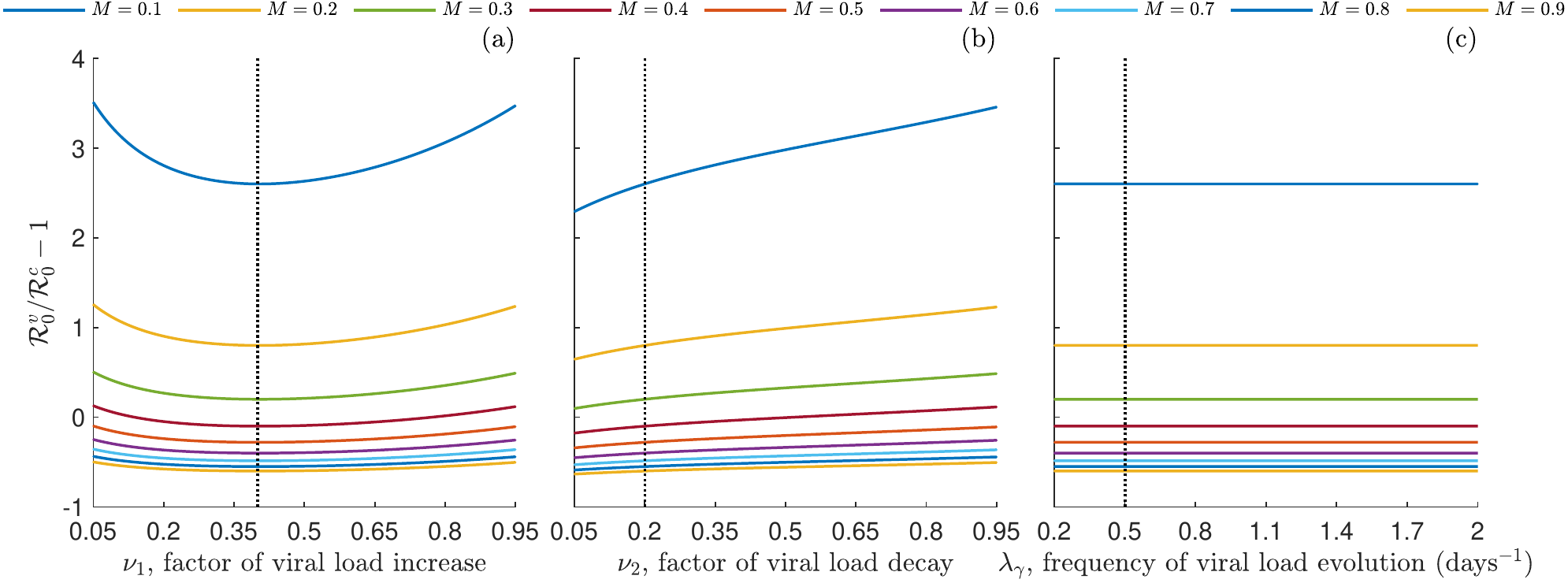}
	\caption{Relative difference of the reproduction number of the model (\ref{macro_simplified}) in the scenario \textbf{S$^v$}, $\mathcal{R}_0^v$, w.r.t. the reproduction number in the scenario \textbf{S$^c$}, $\mathcal{R}_0^c$, for nine values of the factor of transmission normalization, $M\in\{0.1,0.2,0.3,0.4,0.5,0.6,0.7,0.8,0.9\}$ (as indicated in the legend). Panel (a): $\mathcal{R}_0^v/\mathcal{R}_0^c-1$ as a function of the factor of viral load increase, $\nu_1$. Panel (b): $\mathcal{R}_0^v/\mathcal{R}_0^c-1$ as a function of the factor of viral load decay, $\nu_2$. Panel (c): $\mathcal{R}_0^v/\mathcal{R}_0^c-1$ as a function of the frequency of viral load evolution, $\lambda_\gamma$. Black dotted  lines indicate the corresponding baseline value. Other parameters values are given in Table \ref{TabPar}. }\label{fig1}
\end{figure}

An overall view of the relationship between $\mathcal{R}_0^v$ and $\mathcal{R}_0^c$ is provided  
by numerically exploring their mutual position when the relevant model parameters vary in appropriate ranges. 
 In Fig. \ref{fig1}, we display the relative difference of $\mathcal{R}_0^v$ w.r.t. $\mathcal{R}_0^c$, that is the ratio
 \begin{equation*}
 	\label{relativediff}\dfrac{\mathcal{R}_0^v-\mathcal{R}_0^c}{\mathcal{R}_0^c}=\dfrac{\mathcal{R}_0^v}{\mathcal{R}_0^c}-1,
 \end{equation*}
as a function of the factor of viral load increase, $\nu_1$ (Fig. \ref{fig1}a), the factor of viral load decay, $\nu_2$ (Fig. \ref{fig1}b), and the frequency of viral load evolution, $\lambda_\gamma$ (Fig. \ref{fig1}c).  The parameters $\nu_1$, $\nu_2$ and $\lambda_\gamma$ \textit{continuously} vary in the possible ranges of values
\begin{equation}\label{ranges}
	\nu_1,\,\nu_2\in[0.05,0.95],\,\,\lambda_\gamma\in[0.2,2]\,\, \text{days}^{-1}.
\end{equation}
We disregard the extreme cases that $\nu_1,\,\nu_2\approx 0$ and $\nu_1,\,\nu_2\approx 1$, which we consider to be rather unrealistic.  Also, we consider nine values of the factor of transmission normalization $M$ that span the range $[0.1,0.9]$, as indicated in the legend of Fig. \ref{fig1}. The baseline values of the other parameters are those  given in Table \ref{TabPar}. Note that the ratio ${\mathcal{R}_0^v}/{\mathcal{R}_0^c}$ is independent of  $\beta^c$, which does not need to be assigned for the moment.

 From Fig. \ref{fig1}, we observe that the ratio $\mathcal{R}_0^v/\mathcal{R}_0^c$ is a non--monotone convex function of $\nu_1$ (Fig. \ref{fig1}a), an increasing function of $\nu_2$ (Fig. \ref{fig1}b) and an almost constant function of $\lambda_\gamma$ (Fig. \ref{fig1}c), independently of $M$.  In particular, as a function of the factor of viral load increase, $\nu_1$, the ratio $\mathcal{R}_0^v/\mathcal{R}_0^c$ is minimum for intermediate values of $\nu_1$ and assumes almost the same value at $\nu_1=0.05$ and $\nu_1=0.95$.  As regards the irrelevance of  the frequency of viral load evolution $\lambda_\gamma$ on the ratio $\mathcal{R}_0^v/\mathcal{R}_0^c$, it can be explained by the fact that in the current parameter setting the rate of natural death, $\mu$, is much lower than the other parameters contributing to the reproduction numbers. From the expressions (\ref{R0v})--(\ref{R0c}), it follows that  $\mathcal{R}_0^v/\mathcal{R}_0^c$ is almost independent of $\lambda_\gamma$ being the terms multiplied by $\mu$ negligible (see also the point (i) of Remark \ref{rmklimit}). 
 
 From Fig. \ref{fig1} we also observe that 
 the relative difference $\mathcal{R}_0^v/\mathcal{R}_0^c-1$ decreases by increasing the factor of transmission normalization $M$ (see (\ref{betavc})), eventually passing from positive to negative values (namely, from $\mathcal{R}_0^v>\mathcal{R}_0^c$ to $\mathcal{R}_0^v<\mathcal{R}_0^c$). Globally, $\mathcal{R}_0^v$ spans from being about 60\% lower than $\mathcal{R}_0^c$ ($M=0.9$) to  350\% higher than $\mathcal{R}_0^c$ ($M=0.1$).  Also, the sensitivity of $\mathcal{R}_0^v/\mathcal{R}_0^c$ to $M$ greatly diminishes when $M$ overcomes the value 0.5. From Fig. \ref{fig1}a [resp. Fig. \ref{fig1}b] we can note that, for a given value of $M$, the relative difference $\mathcal{R}_0^v/\mathcal{R}_0^c-1$ can pass through zero by varying $\nu_1$ [resp. $\nu_2$], meaning that the mutual position between $\mathcal{R}_0^v$ and $\mathcal{R}_0^c$ changes.  In particular, this happens for $M=0.4$ and it is the reason why we choose it as baseline value for the next numerical investigations in this subsection.
 \begin{figure}[t]\centering
 	\includegraphics[scale=0.9]{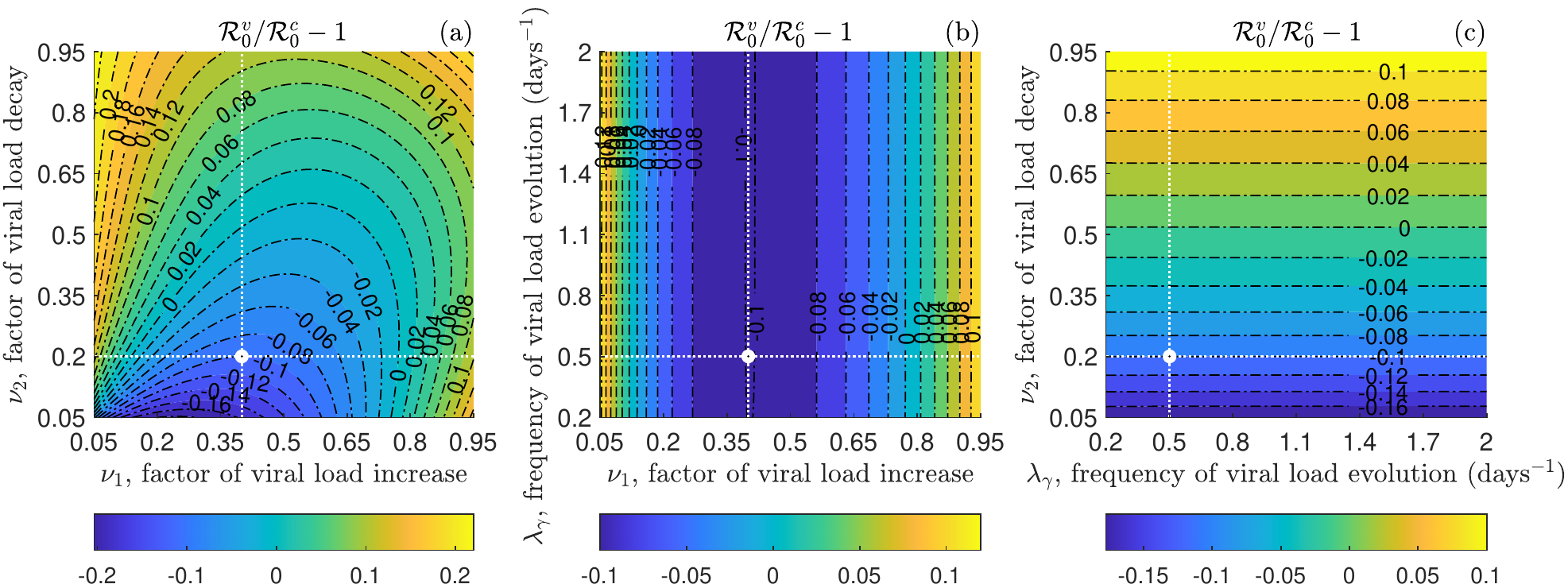}
 	\caption{Counterplots of the relative difference of the reproduction number of the model (\ref{macro_simplified}) in the scenario \textbf{S$^v$}, $\mathcal{R}_0^v$, w.r.t. the reproduction number in the scenario \textbf{S$^c$}, $\mathcal{R}_0^c$. Panel (a):  $\mathcal{R}_0^v/\mathcal{R}_0^c-1$ versus the factor of viral load increase, $\nu_1$, and the factor of viral load decay, $\nu_2$. Panel (b): $\mathcal{R}_0^v/\mathcal{R}_0^c-1$ versus the factor of viral load increase, $\nu_1$, and the frequency of viral load evolution, $\lambda_\gamma$. Panel (c): $\mathcal{R}_0^v/\mathcal{R}_0^c-1$ versus the frequency of viral load evolution, $\lambda_\gamma$, and the factor of viral load decay, $\nu_2$. The intersection between white dotted  lines indicates the corresponding baseline value. Other parameters values are given in Table \ref{TabPar}. }\label{fig2}
 \end{figure}
 \begin{figure}[!ht]\centering
 	\includegraphics[scale=0.95]{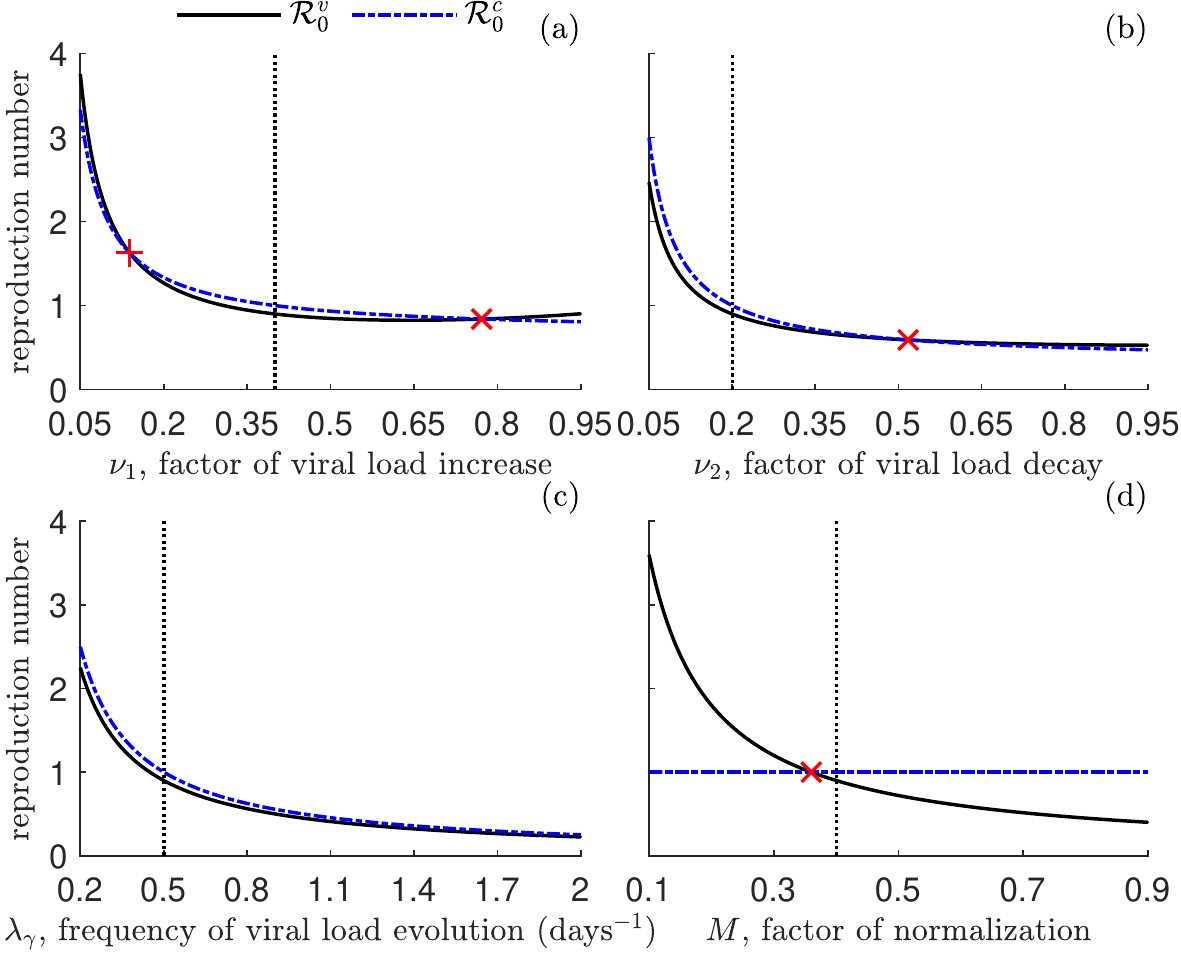}
 	\caption{Reproduction number of the model (\ref{macro_simplified}) in the scenario \textbf{S$^v$}, $\mathcal{R}_0^v$ (black solid lines), and reproduction number in the scenario \textbf{S$^c$}, $\mathcal{R}_0^c$ (blue dash--dotted lines). Panel (a): $\mathcal{R}_0^v$ and $\mathcal{R}_0^c$ as functions of the factor of viral load increase, $\nu_1$. Panel (b): $\mathcal{R}_0^v$ and $\mathcal{R}_0^c$ as functions of the factor of viral load decay, $\nu_2$. Panel (c): $\mathcal{R}_0^v$ and $\mathcal{R}_0^c$ as functions of the frequency of viral load evolution, $\lambda_\gamma$. Panel (d): $\mathcal{R}_0^v$ and $\mathcal{R}_0^c$ as functions of the factor of transmission normalization, $M$.  Black dotted  lines indicate the corresponding baseline values. Red x--marks indicate the intersection points between $\mathcal{R}_0^v$ and $\mathcal{R}_0^c$. Other parameters values are given in Table \ref{TabPar}. }\label{fig3}
 \end{figure}

 In Fig. \ref{fig2}, we display the counterplots of the relative difference $\mathcal{R}_0^v/\mathcal{R}_0^c-1$ as a function of the pairs of parameters $(\nu_1,\nu_2)$, $(\nu_1,\lambda_\gamma)$ and $(\lambda_\gamma,\nu_2)$ in the ranges given in (\ref{ranges}), by setting $M=0.4$. From Fig. \ref{fig2} we can evaluate the combined impact on $\mathcal{R}_0^v/\mathcal{R}_0^c-1$ of two model parameters among $\{\nu_1,\nu_2,\lambda_\gamma\}$ when the third one is set at the baseline value. We can note that, independently of $\lambda_\gamma$, it is $\mathcal{R}_0^v/\mathcal{R}_0^c-1>0$ (that is, $\mathcal{R}_0^v>\mathcal{R}_0^c$) if $\nu_1>0.8$ or $\nu_2>0.6$. Otherwise, if $\nu_1<0.8$ and $\nu_2<0.6$, then $\mathcal{R}_0^v/\mathcal{R}_0^c-1$ can be negative (namely, $\mathcal{R}_0^v<\mathcal{R}_0^c$). In other words, the  reproduction number of  the model (\ref{macro_simplified}) in the scenario \textbf{S$^v$} is greater than the  reproduction number in the scenario \textbf{S$^c$} when at least one between the factor of viral load increase and the factor of viral load decay is high, whilst $\mathcal{R}_0^v$ can be less than $\mathcal{R}_0^c$ when both $\nu_1$ and $\nu_2$ are medium--low. However, in any case the relative difference of $\mathcal{R}_0^v$ w.r.t. $\mathcal{R}_0^c$ is rather small: $\mathcal{R}_0^v$ is at most 20\% greater or smaller than $\mathcal{R}_0^c$. 
 
 The relative difference $\mathcal{R}_0^v/\mathcal{R}_0^c-1$ does not provide information about the exact values assumed by $\mathcal{R}_0^v$ and $\mathcal{R}_0^c$. In order to complete the investigations, in Fig. \ref{fig3} we provide the values of $\mathcal{R}_0^v$ and $\mathcal{R}_0^c$ as a function of one among the parameters $\{\nu_1,\nu_2,\lambda_\gamma\}$ in the ranges (\ref{ranges}) when the other ones are set at the baseline value (Figs. \ref{fig3}a--c), and as a function of the normalization factor $M\in[0.1,0.9]$ when the other parameters are set at the baseline value (Fig. \ref{fig3}d). Here,  we assume  that in the case of constant transmission rate the baseline value of the reproduction number is at the threshold 1, namely $$\mathcal{R}_0^c=1,$$ which yields  $\beta^c=0.10$ days$^{-1}$.
 
 From Fig. \ref{fig3} we note that, for values of $\nu_1$, $\nu_2$ and $\lambda_\gamma$ slightly smaller than the baseline value (see Figs. \ref{fig3}a--c, black dotted  lines), it is $\mathcal{R}_0^c$ above the  threshold 1 (blue dash--dotted  lines) and $\mathcal{R}_0^v$ below the  threshold 1 (black solid  lines). This suggests that, for given epidemiological conditions, the disease dynamics predicted by model (\ref{macro_simplified}) can be radically different depending on the modelling assumption about the transmission function $\nu_\beta(\cdot)$, namely if $p=0$ or $p=1$ in (\ref{macro_simplified}). Also, from Fig. \ref{fig3}d, we note that the reproduction number $\mathcal{R}_0^v$ varies from about 3.6 to 0.4 by varying $M$, by crossing the threshold 1 (that is also the value of $\mathcal{R}_0^c$) for $M\approx 0.36$.

\subsection{Impact of viral load on the disease dynamics}

In this subsection, we numerically explore the impact of different modelling assumptions about the  transmission rate on the disease dynamics predicted by model (\ref{macro_simplified}).
At this aim, we consider the following illustrative epidemiological setting.

 Initial data are set to the beginning of an epidemic, namely there is  a single infectious individual in a totally susceptible population:
\begin{equation}\label{CIvalues}
	\rho_{S,0}=(\bar N-1)/N_{tot},\,\,\rho_{I_1,0}=1/N_{tot},\,\,n_{I_1,0}=v_0\rho_{I_1,0},\,\,\rho_{I_2,0}=n_{I_2,0}=0.
\end{equation}
Here, like in the paper \cite{DMrLnTa22}, we assume that $$\mathcal{R}_0^c=4,$$  which yields $\beta^c=0.4$ days$^{-1}$.

We denote by $t_f$ the terminal time of our numerical simulations (i.e. time horizon). We want that the $t_f$  is a finite time with a reasonable epidemiological interpretation. To this end, inspired by the approach adopted  in the papers \cite{bolzoni2021optimal,hansen11JOMB}, we assume that $t_f$ coincides with  the end of the \textit{first} epidemic wave, namely  $t_f$ is the first time there is less than one infectious individual in the  population:
\begin{equation}
	t_f=\text{inf} \left\{t \in \mathbb{R}_+ \left| \rho_1(t)+\rho_2(t)<\dfrac{1}{N_{tot}} \right.\right\}.\label{tf}
\end{equation}
In other words, $t_f$ is the first time at which  $\rho_1+\rho_2$ drops to $1/N_{tot}$. Of course,  the presence of subsequent epidemic waves is not excluded, but for the sake of simplicity we focus here on just the first one.
\begin{figure}[t]\centering
	\includegraphics[scale=0.95]{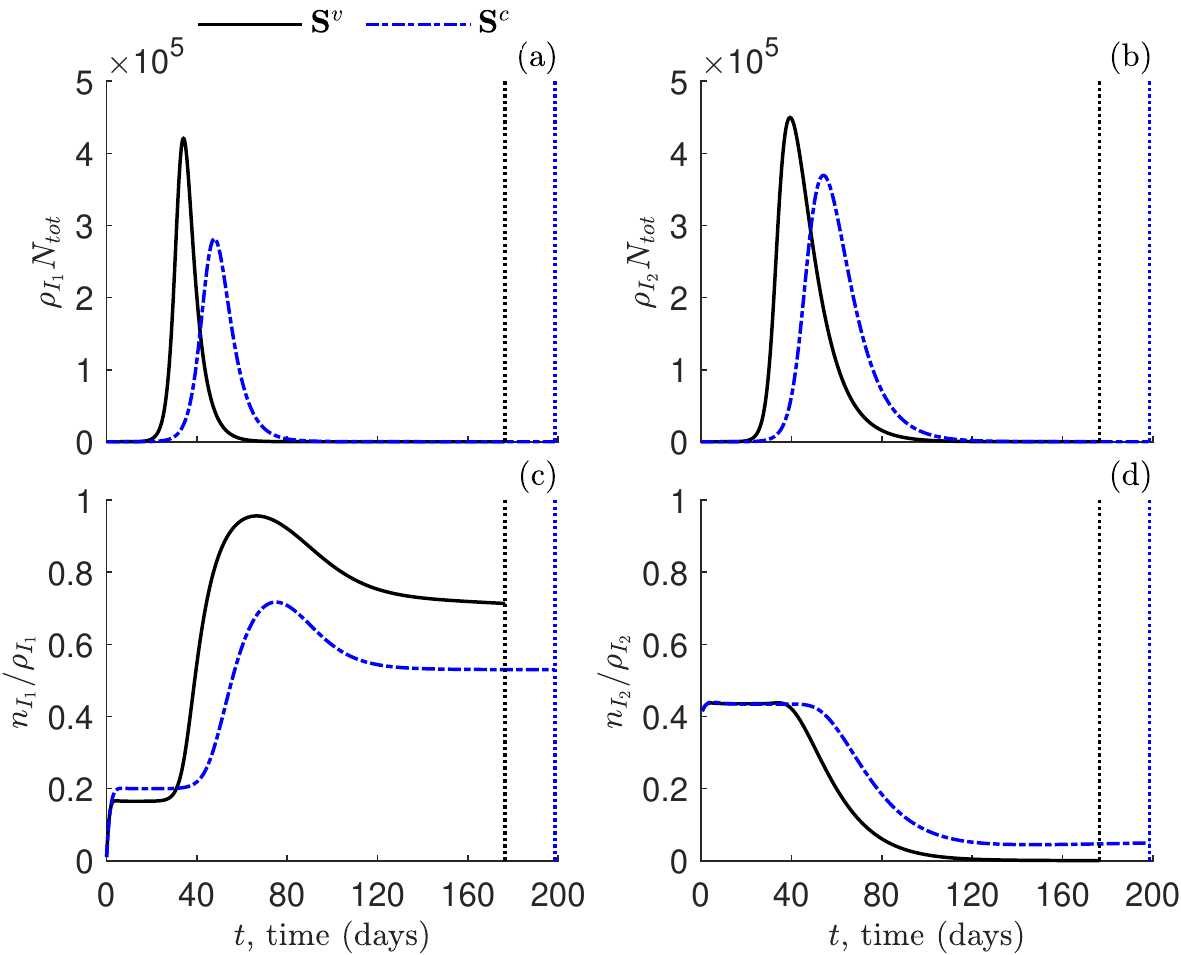}
	\caption{Numerical solutions as predicted by the model  (\ref{macro_simplified}) in scenarios \textbf{S$^v$} (black solid lines) and \textbf{S$^c$} (blue dash--dotted lines). Panel (a): compartment size of infectious individuals with increasing viral load, $I_1$. Panel (b): compartment size  of infectious individuals with decreasing viral load, $I_2$.
		Panel (c): mean viral load of infectious individuals with increasing viral load, $I_1$. Panel (d): mean viral load  of infectious individuals with decreasing viral load, $I_2$. Dotted  lines indicate the corresponding terminal time, as defined in (\ref{tf}). Initial conditions and other parameters values are given in (\ref{CIvalues}) and Table \ref{TabPar}, respectively.  }\label{fig4}
\end{figure}

In order to estimate the factor of transmission normalization $M$, we  consider the model (\ref{macro_simplified})  in the case \textbf{S$^c$}  and denote by $\rho_{I_1}^{c}(t)$, $\rho_{I_2}^{c}(t)$, $n_{I_1}^{c}(t)$, $n_{I_2}^{c}(t)$ the corresponding solutions for $\rho_{I_1}(t)$, $\rho_{I_2}(t)$, $n_{I_1}(t)$, $n_{I_2}(t)$, respectively. Then, $M$ is set to
\begin{equation*}
	M=\dfrac{1}{t_f}\int_0^{t_f}\dfrac{n_{I_1}^c(t)+n_{I_2}^c(t)}{\rho_{I_1}^c(t)+\rho_{I_2}^c(t)}dt,
\end{equation*}
that is the average value of the mean viral load of the infectious population over $[0,t_f]$. In such a way, we obtain
$$M=0.17,$$
yielding
$$\mathcal{R}_0^v=8.47.$$
It turns out that $\mathcal{R}_0^v$ is more than 110\% higher than $\mathcal{R}_0^c$, suggesting that the epidemic wave predicted in the scenario $\mathbf{S}^v$ could be much more devastating than in the scenario  $\mathbf{S}^c$. The other parameter values are given in Table \ref{TabPar}.

In Fig. \ref{fig4}, we display the numerical solutions of model (\ref{macro_simplified}) in the case of viral load--dependent transmission rate, $\mathbf{S}^v$ (black solid lines), and in the case of constant transmission rate, $\mathbf{S}^c$ (blue dash--dotted lines). Specifically, Figs. \ref{fig4}a and \ref{fig4}c [resp. \ref{fig4}b and \ref{fig4}d] report the temporal dynamics of the compartment size of the infectious individuals in $I_1$ [resp. in $I_2$] and the corresponding mean viral load.  Dotted  lines indicate the terminal time (\ref{tf}). 

In accordance with the values of  the reproduction numbers, from Figs. \ref{fig4}a and \ref{fig4}b  we observe that in the case $\mathbf{S}^v$  the  peak of infectious prevalence (i.e., $\max(\rho_{I_1}+\rho_{I_2})$) occurs earlier than in the case $\mathbf{S}^c$ (at day 36 vs at day 51) and it is also higher (about 770,000 vs 600,000 total infectious individuals). However, the epidemic wave ends earlier in the case $\mathbf{S}^v$ ($t_f=176$ days) than in the case $\mathbf{S}^c$ ($t_f=198$ days). Cumulatively, the total number of infections during the epidemic wave are:
\begin{equation*}
	\text{CI}=N_{tot}\int_{0}^{t_f}\beta\left(\dfrac{n_{I_1}(t)+n_{I_2}(t)}{\rho_{I_1}(t)+\rho_{I_2}(t)}\right)^p \rho_S(t)(\rho_{I_1}(t)+\rho_{I_2}(t))dt,
\end{equation*}
where CI stands for \textit{cumulative incidence}. We obtain that in the case  $\mathbf{S}^v$ ($\beta=\beta^v$, $p=1$) about all the initial population becomes infected, whilst in the case $\mathbf{S}^c$ ($\beta=\beta^c$, $p=0$) only 6,000 infections are avoided. From a mathematical point of view, this means that the area under the curve $\rho_{I_1}+\rho_{I_2}$  changes little by varying the modelling assumption about the disease transmission rate. In Table \ref{tab2}, we  collect some relevant epidemiological quantities in the scenarios $\mathbf{S}^v$ and $\mathbf{S}^c$.

{In Figs. \ref{fig4}c and \ref{fig4}d we report the temporal dynamics of the mean viral load of the infectious compartments.  As anticipated in Section \ref{sec:2}, we  highlight that the mean viral load is reliable when the number of particles is sufficiently high due to the law of large numbers (see \cite{DMrLnTa22} for a more detailed discussion).}
From Fig. \ref{fig4}c we note that the mean viral load of the compartment $I_1$ in the case $\mathbf{S}^v$ (black solid line) is -- for most of the time horizon --
higher than in the case $\mathbf{S}^c$ (blue dash--dotted  line). At variance, the mean viral load of the compartment $I_2$ is smaller in the case $\mathbf{S}^v$ w.r.t. the case  $\mathbf{S}^c$ (Fig.  \ref{fig4}d). This suggests that the model (\ref{macro_simplified}) under the assumption of constant disease transmission tends to underestimate the mean viral load of infectious individuals in the increasing phase and to overestimate the mean viral load of infectious individuals in the decreasing phase.
From Figs. \ref{fig4}c and \ref{fig4}d we also note that, {at the  end of the time horizon when the compartments $I_1$ and $I_2$ are almost empty,} the mean viral load remains approximately constant at a positive value, suggesting that the  viral load momentum $n_{I_1}$ [resp. $n_{I_2}$] and the density $\rho_{I_1}$ [resp. $\rho_{I_2}$] go to zero with the same \textit{speed}. 
\begin{table}[t]
	\centering\begin{tabular}{|@{}c|c|c|c|c@{}|}
		\hline
		Scenario &$\max(\rho_{I_1}+\rho_{I_2})N_{tot}$&arg$\max(\rho_{I_1}+\rho_{I_2})$& $t_f$& CI \\
		\hline 
		$\mathbf{S}^v$&  $7.70 \cdot 10^{5}$ & 35.6 days &176.4 days & $1.00\cdot 10^{6}$ \\
		$\mathbf{S}^c$& $5.96 \cdot 10^{5}$& 50.6 days& 198.5 days & $9.94\cdot 10^{5}$ \\
		\hline
	\end{tabular}\caption{Relevant quantities as predicted by the model (\ref{macro_simplified}) in the case of viral load--dependent transmission rate (scenario $\mathbf{S}^v$, first line) and in the case of constant transmission rate  (scenario $\mathbf{S}^c$, second line). First column: infectious prevalence peak, $\max(\rho_{I_1}+\rho_{I_2})N_{tot}$. Second column: time of infectious prevalence peak, arg$\max(\rho_{I_1}+\rho_{I_2})$. Third column: terminal time,  $t_f$. Fourth column: cumulative incidence at $t_f$,  CI. Initial conditions and other parameters values are given in (\ref{CIvalues}) and Table \ref{TabPar}, respectively.}\label{tab2}
\end{table}

\section{Conclusion}\label{sec:6}

{In this work, we  propose and analyse an SIR epidemic model with  viral load--dependent transmission. The  compartmental model is formally derived -- by the means of kinetic equations --  from a stochastic particle description of the individual course of the disease and the viral load progression. This approach allows the macroscopic model to inherit the features of the microscopic dynamics related to the heterogeneity of the viral load in the population  \cite{DMrLnTa22}.
	
	The main results are as follows:
	\begin{itemize}
		\item the particle stochastic model provides that, in the binary interaction between a susceptible and an infectious individual, the probability for the former to get infected depends on the viral load of the latter. In particular, the transmission function is a non--decreasing function of the viral load of the infectious individual. In the macroscopic model, the rate of disease transmission turns out to be a function of the \textit{mean viral load} of the infectious population;
		\item we analytically and numerically investigate the impact of different modelling assumptions about  the  disease transmission rate on the epidemic dynamics. In particular, we consider the case that the  transmission rate linearly depends on the viral load (scenario $\mathbf{S}^v$), which is compared to the classical case of constant  transmission rate (scenario $\mathbf{S}^c$);
		\item we determine explicitly the equilibria of the macroscopic model in the scenario $\mathbf{S}^v$ and study their stability in terms of the basic reproduction number ($\mathcal{R}_0$) of the model.  We prove that a transcritical forward bifurcation occurs at the disease--free equilibrium and  $\mathcal{R}_0=1$. The expression of the $\mathcal{R}_0$ appears to be much more complex than that in the scenario $\mathbf{S}^c$, by depending -- inter alia --  on the initial viral load of the infectious individuals;
		\item the numerical simulations unravel the relationship between the  reproduction numbers $\mathcal{R}_0$ in the scenarios $\mathbf{S}^v$ and $\mathbf{S}^c$ when the model parameters vary in appropriate ranges. We observe that the mutual position between the two $\mathcal{R}_0$'s is almost independent of the frequency of the viral load evolution, but it may be noticeably affected by the factors of viral load increase and decay. Interestingly, for given parameter values, it may be $\mathcal{R}_0<1$ in scenario $\mathbf{S}^v$ and  $\mathcal{R}_0>1$ in scenario  $\mathbf{S}^c$, suggesting that the disease dynamics can be radically different depending on the modelling assumption about the transmission rate;
		\item we simulate an epidemic wave by assuming $\mathcal{R}_0=4$ in the scenario $\mathbf{S}^c$ and estimating the $\mathcal{R}_0$ in the scenario $\mathbf{S}^v$ accordingly.  We obtain that in the case of viral load--dependent transmission, the epidemic wave is more severe, with a higher and earlier prevalence peak, than in the case of constant transmission. Also,  the model  in the scenario $\mathbf{S}^c$ tends to underestimate the mean viral load of infectious individuals whose viral load is  increasing  and to overestimate the mean viral load of infectious individuals whose viral load is decreasing.
	\end{itemize}
The role of viral load in the dynamics of infectious diseases has recently attracted the interest of mathematical epidemiologists. Our work makes a further step in this line of research, by 
  serving as a proof--of--principle verification of the impact of the individuals' viral load on the disease transmission rate and the consequent epidemic dynamics.
 
In the proposed framework, the description of the microscopic mechanisms and the heterogeneity of the viral load at the microscopic level allows one to derive a macroscopic model, which provides for a richer description of the  disease spreading in the host population w.r.t. classical epidemic models. Here we only consider the explicit  influence of the viral load on the transmission mechanism, but, in principle, other switches of individuals between compartments may depend on the viral load at the microscopic level, and on the mean viral load at the macroscopic level.  
Also, more complex situations could be addressed, for example by assuming  different initial viral loads of the  infectious individuals that may give rise to a different epidemic scenario.

We underline that our model does not address a specific infectious disease.
Of course, in presence of exhaustive data concerning the progression of the individuals' viral load, the modelling assumptions can be reformulated or adjusted according on the particular case.
However, even if our model is too simple 
to provide reliable solutions to
real--world epidemics from the quantitative point of view, our
theoretical findings can be used to inform more complex simulation
models developed for specific epidemiological scenarios where more realistic
descriptions of the biological and epidemiological processes are included.	

\medskip
\paragraph*{Acknowledgements} This work was supported by GNFM (Gruppo Nazionale per la Fisica Matematica) of INdAM (Istituto Nazionale di Alta Matematica) through a ``Progetto Giovani 2020'' grant. This work was also partially supported by the Italian Ministry of University and Research (MUR)  
through the ``Dipartimenti di Eccellenza'' Programme (2018-2022), Department of Mathematical Sciences ``G. L. Lagrange'', Politecnico di Torino (CUP: E11G18000350001). This work is also part of the activities of the PRIN 2020 project (No. 2020JLWP23) ``Integrated Mathematical Approaches to Socio--Epidemiological Dynamics''.

\paragraph*{Compliance with Ethical Standards and Competing interests}
The authors declare no conflicts of interest.

\bibliographystyle{abbrv}
\bibliography{DMrLnTa-viral_load_transmission}
\end{document}